\documentclass[11pt]{amsart}
\usepackage{amsmath}
\usepackage{slashed}
\usepackage{graphicx}

\theoremstyle{plain}
\newtheorem{theorem}{Theorem}[section]
\newtheorem{lemma}[theorem]{Lemma}

\newtheorem{cor}[theorem]{Corollary}

\newtheorem{proposition}[theorem]{Proposition}

\theoremstyle{definition}
\newtheorem{remark}[theorem]{Remark}

\numberwithin{equation}{section}

\renewcommand{\div}{\operatorname{div}}
\newcommand{\eps}{\varepsilon}

\newcommand{\vol}{\operatorname{Vol}}

\renewcommand{\div}{\operatorname{div}}

\begin{document}

\title[Proof of the Riemannian Penrose Inequality with Charge for Multiple Black Holes] {Proof of the Riemannian Penrose Inequality with Charge for Multiple Black Holes}

\author[Khuri]{Marcus Khuri}
\address{Department of Mathematics\\
Stony Brook University\\
Stony Brook, NY 11794, USA}
\email{khuri@math.sunysb.edu}

\author[Weinstein]{Gilbert Weinstein}
\address{Physics Department and Department of Computer Science and Mathematics\\
Ariel University of Samaria\\
Ariel, 40700, Israel}
\email{gilbert.weinstein@gmail.com}

\author[Yamada]{Sumio Yamada}
\address{Department of Mathematics\\
Gakushuin University\\
Tokyo, 171-8588, Japan}
\email{yamada@math.gakushuin.ac.jp}

\thanks{M. Khuri acknowledges the support of
NSF Grants DMS-1007156 and DMS-1308753. S. Yamada acknowledges the support of
JSPS Grants 23654061 and 24340009.}

\begin{abstract}
We present a proof of the Riemannian Penrose inequality with charge in the context of asymptotically flat initial data sets for the Einstein-Maxwell equations, having possibly multiple black holes with no charged matter outside the horizon, and satisfying the relevant dominant energy condition. The proof is based on
a generalization of Hubert Bray's conformal flow of metrics adapted to this setting.
\end{abstract}
\maketitle

\section{Introduction}
\label{sec1}

In a seminal paper \cite{Penrose} (see also \cite{Penrose1}), in which he proposed the celebrated cosmic censorhip conjecture, R. Penrose also proposed a related inequality, now referred to as the Penrose Inequality. The inequality
is derived from cosmic censorship via a heuristic argument relying on Hawking's area
theorem \cite{HawkingEllis}. Consider an asymptotically flat Cauchy surface in a spacetime satisfying
the dominant energy condition, having ADM mass $m$, and containing an event horizon of area $A=4\pi \rho^2$, which undergoes gravitational collapse and settles to a Kerr-Newman solution. Since the ADM mass $m_\infty$ of the final state is no greater
than $m$, the area radius $\rho_\infty$ is no less than $\rho$, and the final state must satisfy
$m_\infty\geq \frac12 \rho_\infty$ in order to avoid naked singularities, it must have been the case that $m\geq \frac12 \rho$ also at the
beginning of the evolution. A counterexample to the Penrose inequality would therefore suggest data which leads
under the Einstein evolution to naked singularities, and a proof of the Penrose inequality may be viewed as evidence in
support of cosmic censorship.

The event horizon is indiscernible in the original slice without knowing the full evolution, however one may, without
disturbing this inequality, replace the event horizon by the outermost minimal area enclosure of the apparent horizon (the boundary of the region admitting
trapped surfaces). The inequality further simplifies in the time-symmetric case, in which the outermost minimal area enclosure of the apparent horizon coincides with the outermost minimal surface, and the
dominant energy condition reduces simply to nonnegative scalar curvature. This leads to the Riemannian version of the
inequality: the ADM mass $m$ and the area radius $r$ of the outermost minimal surface in an asymptotically flat
3-manifold of nonnegative scalar curvature, satisfy
\begin{equation} \label{penrose-inequality}
  m\geq \frac{\rho}{2}
\end{equation}
with equality if and only if the
manifold is isometric to the canonical slice of the Schwarzschild spacetime. Note that this characterizes the canonical slice of Schwarzschild as the unique minimizer of
$m$ among all such 3-manifolds admitting an outermost horizon of area $A=4\pi \rho^2$.

This inequality was first proved in the special case where the horizon is connected by
Huisken and Ilmanen \cite{HuiskenIlmanen} using the inverse mean curvature flow, an approach proposed by
Jang and Wald \cite{JangWald}, following Geroch \cite{Geroch} who had shown that the Hawking mass is
nondecreasing under the flow. The inequality was proven in full generality by Bray \cite{Bray} using a conformal flow of the initial Riemannian metric, and the positive mass theorem \cite{SchoenYau}, \cite{Witten}.

We now turn to the charged case which is somewhat more subtle. It is natural to conjecture as above that
the Reissner-Nordstr\"{o}m spacetime, the charged analog of Schwarzschild, is the unique minimizer of $m$, given $\rho$ and $q$. Since Reissner-Nordstr\"{o}m
satisfies $m=\frac12(\rho+q^2/\rho)$ where $q$ is the total charge, one is thus lead to conjecture that in any asymptotically flat data
satisfying an appropriate energy condition it holds
\begin{equation} \label{charged-penrose-inequality}
  m \geq \frac12 \left( \rho + \frac{q^2}{\rho} \right),
\end{equation}
with equality if and only if the initial data is the canonical slice of Reissner-Nordstr\"{o}m. This follows from \cite{HuiskenIlmanen}, and is based on Jang \cite{Jang}, but only for a connected horizon, since the proof relies on inverse mean curvature flow. In
fact \eqref{charged-penrose-inequality} can fail if the horizon is not connected, and a counterexample based on the Majumdar-Papapetrou
spacetime with two black holes was constructed in \cite{WeinsteinYamada}.
This counterexample nonetheless does not suggest a counterexample to cosmic censorship. This is because the
right-hand side of \eqref{charged-penrose-inequality} is not monotone increasing in $\rho$. Indeed, already Jang observed
that \eqref{charged-penrose-inequality} is equivalent to two inequalities:
\begin{equation} \label{upper-lower-bound}
  m - \sqrt{m^2 - q^2} \leq \rho \leq m + \sqrt{m^2-q^2}.
\end{equation}
Cosmic censorship suggests that the upper bound always holds, while the counterexample in \cite{WeinsteinYamada} violates the lower bound. It turns out, however, that the lower bound also
holds, and furthermore is motivated by cosmic censorship in the case of a single black hole, or more generally when $\rho\geq|q|$ (see \cite{DainKhuriWeinsteinYamada}).

In this paper, we prove the upper bound in \eqref{upper-lower-bound} for multiple black holes. By the positive mass
theorem with charge, $m\geq
|q|$ with equality if and only if the data is Majumbdar-Papapetrou \cite{GibbonsHawkingHorowitzPerry}; see \cite{ChruscielReallTod}, \cite{KhuriWeinstein} for the rigidity result. Hence if $\rho \leq |q|$, the
upper
bound in \eqref{upper-lower-bound} follows immediately
\begin{equation}
  \rho\leq |q| \leq m \leq m + \sqrt{m^2-q^2}.
\end{equation}
It thus only remains to prove the upper bound under the additional hypothesis $|q| < \rho$. Under this hypothesis, it
is the lower bound that follows immediately
\begin{equation}
  m \leq |q| + \sqrt{m^2-q^2} < \rho + \sqrt{m^2-q^2}.
\end{equation}
In fact the condition $|q| \leq \rho$ is always valid for a single horizon, in light of its stability \cite{Gibbons}, \cite{KhuriWeinsteinYamada}, however for multiple horizons this
inequality is indeed a nontrivial restriction. In view of all the above, the upper
bound in \eqref{upper-lower-bound} is equivalent to \eqref{charged-penrose-inequality} under the additional hypothesis
$|q|\leq \rho$. The proof of this latter statement will be based on a generalization of Bray's conformal flow. It should also be noted that the right-hand side of \eqref{charged-penrose-inequality} is nondecreasing as
a function of $\rho$ (with fixed $q$), precisely when $|q|\leq \rho$. Thus, \eqref{charged-penrose-inequality} with the auxiliary area-charge inequality may also be derived using the heuristic Penrose argument.

\begin{figure} \label{graph}
\includegraphics{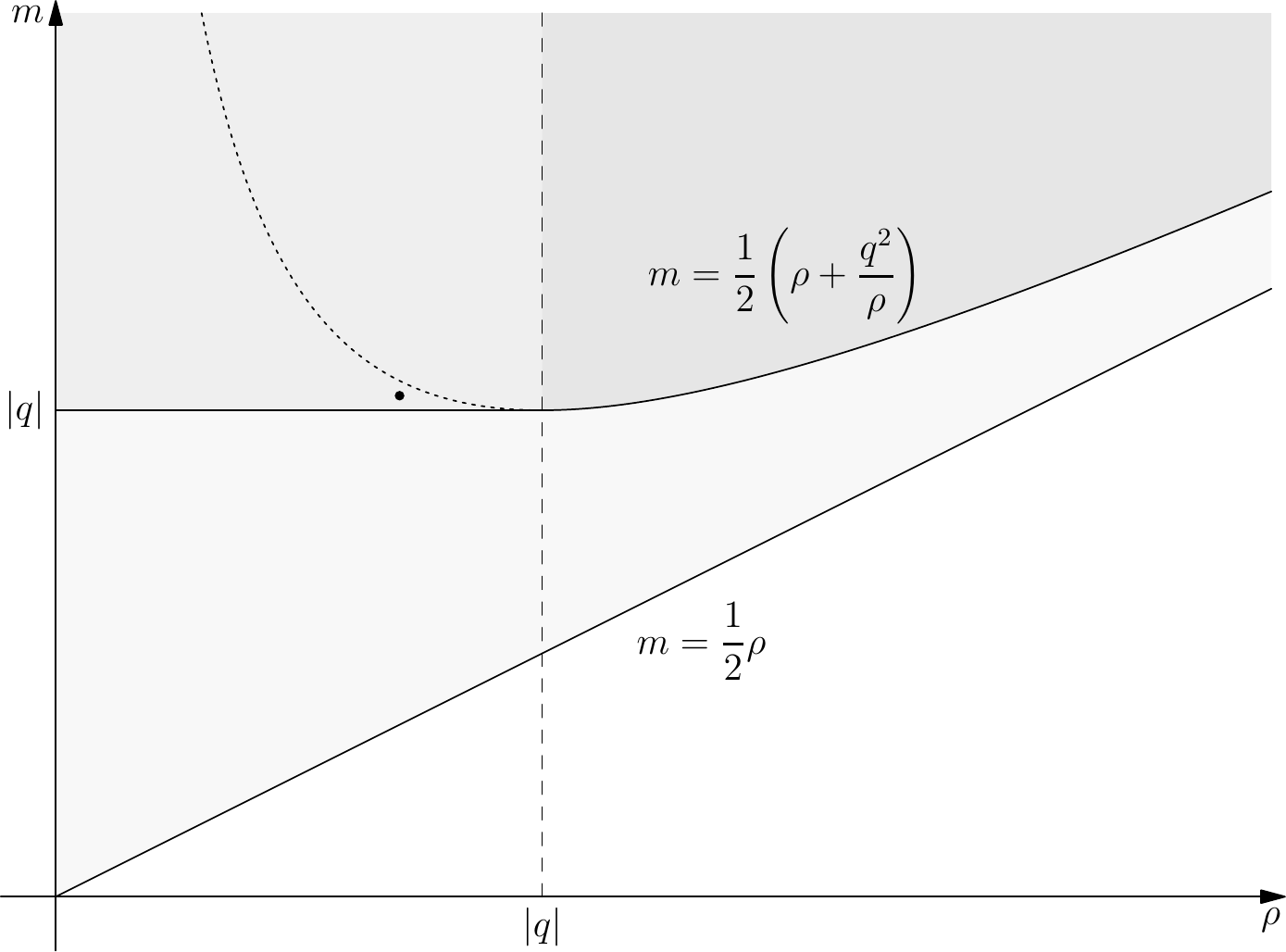}
 \caption{Graphical representation of geometric inequalities}
\end{figure}

The inequalities discussed in the previous paragraphs are most easily visualized in Figure \ref{graph}.
The white area is the positive mass theorem $m\geq0$.
The light shaded area is the Penrose inequality, and the other two darker shaded areas are the charged Penrose inequality.
The inequality represented by the region to the left of the dashed vertical
line $\rho=|q|$, and above the solid horizontal line $m=|q|$, follows from the charged positive mass theorem. Moreover, the dotted curve in this region is the lower bound
in \eqref{upper-lower-bound}, or a continuation of the equality curve from \eqref{charged-penrose-inequality}; the black dot represents the counterexample in \cite{WeinsteinYamada}.
This paper deals with the proof of the inequality represented by the darkest shaded
region, to the right of $\rho=|q|$ and above $m=\frac{1}{2}(\rho+q^2/\rho)$, which is the upper bound in \eqref{upper-lower-bound} with $|q|\leq\rho$, or \eqref{charged-penrose-inequality} with $|q|\leq\rho$. Lastly,
it should be noted that every configuration with one black hole component lies to the right of the vertical dashed line $\rho=|q|$, or equivalently, every configuration to the left of the vertical line has multiple black hole components.

We end the introduction with a few definitions and the statement of our main theorem and its corollaries. An initial
data set $(M,g,E,B)$ consists of a $3$-manifold $M$, a Riemannian metric $g$, and vector fields $E$ and $B$. It will be assumed that the data satisfy the Maxwell constraints with no charges outside the horizon $\div_g E=\div_g B=0$, and that the charged dominant energy condition
\begin{equation} \label{dec}
  16\pi\mu_{EM}=R_{g}- 2(|E|_{g}^2+|B|_{g}^2)\geq 0
\end{equation}
is valid, where $R_{g}$ is the scalar curvature of $g$ and $\mu_{EM}$ is the energy density of the matter fields after contributions from the electromagnetic field have been removed. It should be noted that typically the charged dominant energy condition is given by the slightly stronger statement $\mu_{EM}\geq|J_{EM}|_{g}$, where
$4\pi J_{EM}=E\times B$ is minus one half the momentum density of the electromagnetic field. It turns out, however, that for the results of the current paper the hypothesis \eqref{dec} is sufficient. Moreover in the case of equality for \eqref{charged-penrose-inequality}, it will be shown that $E$ and $B$ are linearly dependent so that $J_{EM}=0$. Typically when Penrose-type inequalities are saturated, the vanishing of the momentum density arises at least in part due to the stronger version of the charged dominant energy condition. Nevertheless, the same result holds here under the weaker form of the energy condition \eqref{dec}.

We assume further that
the data is strongly asymptotically flat, meaning that there is a compact set $K$ such that $M\setminus K$ is the finite union of disjoint ends, and in the coordinates given on each end the fields decay according to
\begin{equation} \label{saf}
  g_{ij}=\delta_{ij} +O_{2}(|x|^{-1}), \quad E_{i} = O_{1}(|x|^{-2}), \quad B_{i} = O_{1}(|x|^{-2}).
\end{equation}
and $R_g$ is integrable. This guarantees that the ADM mass and the total electric and magnetic
charges
\begin{gather}
  \label{adm}
  m = \frac1{16\pi} \int_{S_\infty} (g_{ij,j}-g_{jj,i}) \nu^i\, dA, \\
  \label{charges}
  q_e = \frac1{4\pi} \int_{S_\infty} E_i \nu^i\, dA, \quad
  q_b = \frac1{4\pi} \int_{S_\infty} B_i \nu^i\, dA,
\end{gather}
are well defined, with squared total charge $q^2=q_{e}^{2}+q_{b}^{2}$. Here $\nu$ is the outer unit normal, and the limit is taken in a designated end.
Without loss of generality, we assume that the magnetic charge $q_b=0$, and so from now on $q=q_e$. This can
always be achieved by a fixed rotation in $(E,B)$ space. Conformally compactifying all but the designated end, we
can now restrict our attention to surfaces which bound compact regions, and define $S_2$ to enclose $S_1$ to mean
$S_1=\partial K_1$, $S_2=\partial K_2$ and $K_1\subset K_2$. An outermost horizon is a compact minimal surface not enclosed in any other compact minimal surface. The following results were first discussed in the announcement \cite{KhuriWeinsteinYamada1}.

\begin{theorem} \label{main:theorem}
Let $(M,g,E,B)$ be a strongly asymptotically flat initial data set with outermost minimal surface boundary of area $A=4\pi \rho^2$, satisfying the charged dominant energy condition and the Maxwell constraints without charged matter. If $|q|\leq \rho$, then \eqref{charged-penrose-inequality} holds with equality if and only if the data set arises as the canonical slice of the Reissner-Nordstr\"{o}m spacetime.
\end{theorem}

\begin{cor} \label{main:corollary}
Let $(M,g,E,B)$ be a strongly asymptotically flat initial data set with outermost minimal surface boundary of area $A=4\pi \rho^2$, satisfying the charged dominant energy condition and the Maxwell constraints without charged matter. Then the upper bound
in~\eqref{upper-lower-bound} holds with equality if and only if the data set arises as the canonical slice of the Reissner-Nordstr\"{o}m spacetime.
\end{cor}

\begin{cor}
Assume that the above hypotheses hold. If $q$ and $\rho$ are fixed with $|q|\leq \rho$, then the canonical Reissner-Nordstr\"{o}m slice is the unique minimizer of $m$. Moreover,
if $m$ and $q$ are fixed with $m\geq|q|$, then the canonical Reissner-Nordstr\"{o}m slice is the unique maximizer of $\rho$.
\end{cor}

In the case when charged matter is present, and in particular is not compactly supported, counterexamples exist \cite{KhuriWeinsteinYamada2}.
The full version of the inequality in the non-time-symmetric case remains an open problem. A reduction argument similar to that proposed by Bray and the first author in \cite{BrayKhuri1}, \cite{BrayKhuri2}, has been given in \cite{DisconziKhuri} (see also \cite{Khuri}). However it only applies to the case of a single black hole, as it is based on a coupling of the static Jang equation with inverse mean curvature flow. Coupling the static Jang equation to Bray's conformal flow is possible and also leads to a reduction argument for the Penrose inequality; this was briefly discussed in \cite{BrayKhuri2}. It seems likely then that a coupling to the charged conformal flow presented in this paper, should reduce the general charged Penrose inequality to the time-symmetric case as well. Whether the coupled system admits a solution with the appropriate boundary and asymptotic behavior is then an important open question.

This paper is organized as follows. In the next two sections, a generalized version of Bray's conformal flow will be defined, and its existence will be established. In Section \ref{sec4}, it will be shown that
the flowing outermost minimal surfaces move out into the asymptotic end and eventually exhaust the manifold. In Bray's original flow this exhaustion always occurs, however for the charged conformal flow,
the exhaustion can only happen when $|q|\leq \rho$, and is one of the most interesting and surprising differences between this flow and the original. Section \ref{sec5} is dedicated to monotonicity of the
mass, which follows from a modified doubling argument in analogy to the original flow. In Section \ref{sec6} we solve a quasi-linear elliptic equation, whose solution plays an important part in the
proof of monotonicity, and in Section \ref{sec7} proofs of the main theorem will be given. It should be pointed out that another difference between the strategy here, and that of Bray for the uncharged inequality, is that our proof reduces the case of multiple horizon components to the case of one component, whereas Bray's proof does not rely on knowledge of the single component case.  Lastly,
two appendices are added which include an auxiliary a priori estimate, and the model example for the new flow.

\section{The Charged Conformal Flow}
\label{sec2} 

The goal here is to construct a flow $(M_{t},g_{t},E_{t},B_{t})$ of asymptotically flat initial data for the Einstein-Maxwell equations, starting from the given initial data $(M,g,E,B)$ at $t=0$, and which preserves the boundary area $|\partial M_{t}|_{g_{t}}$, total charge $q_{t}$, Maxwell
constraints, the charged dominant energy condition, and exhibits a nonincreasing ADM mass $m(t)$. Moreover, this flow should reduce to Bray's conformal flow when $|E|_{g}=|B|_{g}=0$, and should proceed by coordinate rescalings in the standard initial data for Reissner-Nordstr\"{o}m. This flow, defined below, will be referred to as the \textit{charged conformal flow}.

Consider the conformal flow of metrics defined by $g_{t}=u_{t}^{4}g$, with $u_{0}\equiv 1$. Given
the metric $g_{t}$, define $\partial M_{t}$ to be the outermost minimal area enclosure of $\partial M$
in $(M,g_{t})$, and denote the region enclosed by $\partial M_{t}$ and spatial infinity by $M_{t}$. It will turn out that $\partial M_{t}$ does not intersect $\partial M$, and hence it is
an outermost minimal surface. Also set $E_{t}^{i}=u_{t}^{-6}E^{i}$ and $B_{t}^{i}=u_{t}^{-6}B^{i}$.
Given $g_{t}$, $E_{t}$, $B_{t}$, and $\partial M_{t}$, define $v_{t}$ to be the unique solution
of the Dirichlet problem
\begin{equation}\label{2.3}
\Delta_{g_{t}}v_{t}-\left(|E_{t}|_{g_{t}}^{2}+|B_{t}|_{g_{t}}^{2}\right)v_{t}=0,\text{ }\text{ }\text{ }\text{ on }\text{ }\text{ }\text{ }M_{t},
\end{equation}
\begin{equation}\label{2.4}
v_{t}=0\text{ }\text{ }\text{ }\text{ on }\text{ }\text{ }\text{ }\partial M_{t},
\text{ }\text{ }\text{ }\text{ }\text{ }\text{ }v_{t}\rightarrow -1 \text{ }\text{ }\text{ }\text{ as }\text{ }\text{ }\text{ }r\rightarrow\infty.
\end{equation}
By expanding the solution in spherical harmonics, it follows that
\begin{equation}
v_{t}= -1+\frac{\gamma_{t}}{r}+O\left(\frac{1}{r^{2}}\right) \text{ }\text{ }\text{ }\text{ as }\text{ }\text{ }\text{ }r\rightarrow\infty,
\end{equation}
for some constant $\gamma_{t}>0$. On $M\setminus M_{t}$ set $v_{t}\equiv 0$. The function $v_{t}$ will act as the logarithmic velocity of the flow $\frac{d}{dt}u_{t}=v_{t}u_{t}$. Thus given $v_{t}$, define $u_{t}=\exp\left(\int_{0}^{t}v_{s}ds\right)$.

The existence and regularity of this flow is similar to that of the original conformal flow, and will be discussed in the next section. Moreover, it is clear that it reduces to Bray's flow when the electromagnetic field vanishes, and indeed is trivial in the Reissner-Nordstr\"{o}m solution as is shown in Appendix B. We now prove that it satisfies the other desired properties.

\begin{theorem}\label{flowproperties}
For all $t\geq 0$ it holds that $q_{t}=q$, $|\partial_{t} M|_{g_{t}}=|\partial M|_{g}$, $\div_{g_{t}}E_{t}=\div_{g_{t}}B_{t}=0$ and
\begin{equation}\label{2.5}
R_{g_{t}}\geq2\left(|E_{t}|_{g_{t}}^{2}+|B_{t}|_{g_{t}}^{2}\right).
\end{equation}
\end{theorem}

\begin{proof}
The same arguments used by Bray \cite{Bray} apply to show that the area remains constant throughout
the flow (see Section \ref{sec3} below). In order to show that the charge remains constant, observe that $|E_{t}|_{g_{t}}^{2}=u_{t}^{-8}|E|_{g}^{2}$, and hence
\begin{equation}\label{6.5}
4\pi q_{t}=\int_{S_{\infty}}g_{t}(E_{t},\nu_{t})dA_{g_{t}}
=\int_{S_{\infty}}g(E,\nu)dA_{g}
=4\pi q.
\end{equation}
Furthermore
\begin{equation}\label{2.6}
\div_{g_{t}}E_{t}=\frac{1}{\sqrt{\det g_{t}}}\partial_{i}\left(\sqrt{\det g_{t}}E_{t}^{i}\right)=\frac{u_{t}^{-6}}{\sqrt{\det g}}\partial_{i}\left(\sqrt{\det g}E^{i}\right)=u_{t}^{-6}\div_{g}E=0,
\end{equation}
and similarly for the magnetic field $B$.

It remains to show that the charged dominant energy condition remains preserved throughout the flow.
Let $L_{g}$ denotes the conformal Laplacian, then by a standard formula
\begin{equation}\label{2.9}
u_{t}^{5}R_{g_{t}}=-8L_{g}u_{t}=-8\left(\Delta_{g}u_{t}-\frac{1}{8}R_{g}u_{t}\right),
\end{equation}
so that with help from the conformal covariance of $L_{g}$ it follows that
\begin{align}\label{4.5}
\begin{split}
\frac{d}{dt}(u_{t}^{8}R_{g_{t}})&=\frac{d}{dt}[u_{t}^{3}(u_{t}^{5}R_{g_{t}})]\\
&=3u_{t}^{2}\left(\frac{d}{dt}u_{t}\right)u_{t}^{5}R_{g_{t}}+u_{t}^{3}\frac{d}{dt}(-8L_{g}u_{t})\\
&=3v_{t}u_{t}^{8}R_{g_{t}}-8u_{t}^{3}L_{g}(u_{t}v_{t})\\
&=3v_{t}u_{t}^{8}R_{g_{t}}-8u_{t}^{8}L_{g_{t}}v_{t}\\
&=3v_{t}u_{t}^{8}R_{g_{t}}-8u_{t}^{8}\left(\Delta_{g_{t}}v_{t}-\frac{1}{8}R_{g_{t}}v_{t}\right)\\
&=4v_{t}u_{t}^{8}R_{g_{t}}-8v_{t}u_{t}^{8}(|E_{t}|_{g_{t}}^{2}+|B_{t}|_{g_{t}}^{2}).
\end{split}
\end{align}
Then since $\frac{d}{dt}(u_{t}^{8}|E_{t}|_{g_{t}}^{2})=\frac{d}{dt}(u_{t}^{8}|B_{t}|_{g_{t}}^{2})=0$, we have
\begin{equation}\label{4.6}
\frac{d}{dt}[u_{t}^{8}\left(R_{g_{t}}-2|E_{t}|_{g_{t}}^{2}-2|B_{t}|_{g_{t}}^{2}\right)]
=4v_{t}u_{t}^{8}\left(R_{g_{t}}-2|E_{t}|_{g_{t}}^{2}-2|B_{t}|_{g_{t}}^{2}\right),
\end{equation}
so that
\begin{equation}\label{4.6.1}
u_{t}^{8}\left(R_{g_{t}}-2|E_{t}|_{g_{t}}^{2}-2|B_{t}|_{g_{t}}^{2}\right)
=e^{\int_{0}^{t}4v_{s}ds}\left(R_{g}-2|E|_{g}^{2}-2|B|_{g}^{2}\right)
=u_{t}^{4}\left(R_{g}-2|E|_{g}^{2}-2|B|_{g}^{2}\right)\geq 0.
\end{equation}
\end{proof}

Monotonicity of the mass is of course more difficult and relegated to its own section,
Section \ref{sec5}. Notice also that we do not prove that the flow converges to the
canonical Reissner-Nordstr\"{o}m data, in analogy with the fact that the original conformal
flow converges to the canonical Schwarzschild data. While we strongly believe that this
result holds for the charged conformal flow, it is not needed to prove the main theorem
and is hence left for future investigation.

\section{Existence of the Flow}
\label{sec3}

In this section we prove that the charged conformal flow exists, by employing
the same discretization procedure developed Bray. The presentation will closely follow that in \cite{Bray}.
For each $\epsilon \in (0, \frac12)$ a family
of approximate solutions $u^{\epsilon}_{t} (x)$ will easily be constructed, and the solution
shall arise from the limit
\begin{equation}
u_t(x) = \lim_{\epsilon \rightarrow 0} u^{\epsilon}_t (x).
\end{equation}
Given the metric $g^{\epsilon}_t = (u^{\epsilon}_t)^4 g$ (with $u_0^{\epsilon} \equiv 1$),
define for $t\geq 0$
\begin{equation}
\partial M_{t}^\epsilon  = \left\{ \begin{array}{ll}
         \partial M & \mbox{if $t=0$,} \\
        \mbox{the outermost minimal area enclosure} & \\
         \mbox{of $\partial M_{t-\epsilon}^\epsilon$ in $(M, g_t^\epsilon)$} & \mbox{if
         $t=k\epsilon$ with $k\in\mathbb{Z}_+$,} \\
        \partial M_{\lfloor t \rfloor_\epsilon}^\epsilon  &  \mbox{otherwise,}
        \end{array}\right.
\end{equation}
where
\begin{equation}
{\lfloor t \rfloor_\epsilon} := \epsilon {\left\lfloor \frac{t}{\epsilon} \right\rfloor}.
\end{equation}
Let $M_{t}^{\epsilon}$ denote the region enclosed between $\partial M_{t}^{\epsilon}$ and spatial infinity. Moreover, given $\partial M_{t}^{\epsilon}$ we may define
\begin{equation}\label{defu}
u_t^\epsilon (x) = \exp \Big(\int_0^t v_s^\epsilon (x) \,\, ds \Big),
\end{equation}
where $v_t^\epsilon$ is the solution of the Dirichlet problem
\begin{equation}
\left\{ \begin{array}{rcl}
    \Delta_{g_{\lfloor t \rfloor_\epsilon}^\epsilon} v_t^\epsilon
    - \left(|E_{\lfloor t \rfloor_\epsilon}^\epsilon|_{g_{\lfloor t \rfloor_\epsilon}^\epsilon}^2
    +|B_{\lfloor t \rfloor_\epsilon}^\epsilon|_{g_{\lfloor t \rfloor_\epsilon}^\epsilon}^2\right)  v_t^\epsilon &=& 0 \> \> \> \> \> \> \> \> \> \text{on $M_{t}^\epsilon$,}\\
    v_t^\epsilon &=& 0 \> \> \> \> \> \> \> \> \> \text{on $\partial M_{t}^\epsilon$,}\\
    v_t^\epsilon &\rightarrow&-1 \> \> \> \> \> \>  \text{as $|x|\rightarrow\infty$,}
    \end{array} \right.
\end{equation}
with $v_{t}^\epsilon (x) \equiv 0$ on $M\setminus M_{t}^{\epsilon}$ and $(E_t^\epsilon)^{j}=(u_t^\epsilon)^{-6} E^{j}$, $(B_t^\epsilon)^{j}=(u_t^\epsilon)^{-6} B^{j}$.
Note that \eqref{defu} directly implies that $u^\epsilon_t (x)\rightarrow e^{-t}$ as $|x|\rightarrow\infty$.

Now observe that $\partial M_{t}^\epsilon$ and hence $v_{t}^\epsilon(x)$ are fixed for $t \in [k \epsilon, (k+1)\epsilon)$.  Furthermore, for $t =k \epsilon$ with $k \in\mathbb{Z}_+$, $\partial M_{t}^\epsilon$ does not touch $\partial M_{t-\epsilon}^\epsilon$ because $\partial M_{t-\epsilon}^\epsilon$ has negative mean curvature in $(M, g_t^\epsilon)$. This follows from the fact that $\partial_{\nu}u_{t}^{\epsilon}|_{\partial M_{t-\epsilon}^{\epsilon}}<0$, where $\nu$ is the unit outer normal pointing to spatial infinity. To see that this is in fact the case, first observe that $\partial_{\nu}u_{t-\epsilon}^\epsilon |_{\partial M_{t-\epsilon}^\epsilon}= 0$ since $\partial M_{(k-1)\epsilon}^\epsilon$ is minimal in $(M, g_{(k-1)\epsilon}^\epsilon)$, and $\partial_{\nu}v_{t-\epsilon}^\epsilon |_{\partial M_{t-\epsilon}^\epsilon}<0$ from the Hopf lemma.
Therefore, using $u_{t}^{\epsilon}=u_{t-\epsilon}^\epsilon \exp\Big(\epsilon v_{(k-1)\epsilon}^\epsilon\Big)$ we find that
\begin{equation}
  \partial_{\nu}u_{t}^\epsilon\Big|_{\partial M_{t-\epsilon}^\epsilon}
= u_{t-\epsilon }^\epsilon
\partial_{\nu} \exp \Big( \epsilon v_{(k-1)\epsilon}^\epsilon (x) \Big)\Big|_{\partial M_{t-\epsilon}^\epsilon}
= \epsilon u_{t}^\epsilon \partial_{\nu}v_{t-\epsilon}^\epsilon \Big|_{\partial M_{t-\epsilon}^\epsilon} < 0.
\end{equation}
This inequality says that by pushing the surface $\partial M_{t-\epsilon}^\epsilon$ outwards, the area can be reduced in $(M, g^\epsilon_t)$. Hence, $\partial M_{t-\epsilon}^{\epsilon}$ acts as a barrier in
$(M, g_t^\epsilon)$.  As the outermost condition implies the outer-minimizing condition, $\partial M_{t}^\epsilon$ is actually a strictly outer minimizing horizon of $(M, g_t^\epsilon)$, and is smooth since $g_t^\epsilon$ is smooth outside $\partial M_{t-\epsilon}^\epsilon$.

The same arguments presented in \cite{Bray} yield the following facts. Not only are the surfaces $\partial M_{t}^\epsilon$ smooth, but any limits of these surfaces are smooth.
Furthermore from the definition of $\partial_{t}^\epsilon$, it is apparent that for $\epsilon>0$ the horizon $\partial M_{t_{2}}^\epsilon$ encloses $\partial M_{t_{1}}^\epsilon$ for all $t_2 \geq t_1 \geq 0$. Also, the horizon $\partial M_{t}^\epsilon$ is the outermost minimal area enclosure of $\partial M$ in $(M, g_t^\epsilon)$ when $t = k \epsilon$ with $k \in \mathbb{Z}_+$.

\begin{lemma}
The functions $u_t^\eps (x)$ are positive, bounded, locally Lipschitz functions (in $x$ and $t$) with uniform Lipschitz constants independent of $\epsilon$.
\end{lemma}

\begin{proof}
Positivity is obvious from the definition of $u_t^\epsilon$.  By the maximum principle, $v_t^\epsilon$ cannot achieve a nonnegative maximum.  This then implies that $u_t^\epsilon (x) \leq 1$.  That $u_t^\epsilon(x)$ is Lipschitz in $t$ follows from its definition and the fact that $-1 < v_t^\eps(x) \leq 0$.  That $u_t^\epsilon (x)$ is Lipschitz in $x$ follows from the fact that $v_t^\epsilon(x)$ is Lipschitz in $x$ (with Lipschitz constant depending on $t$), which follows from Corollary 15 of \cite{Bray}.
\end{proof}

\begin{cor}
There exists a subsequence $\{\epsilon_i\}$ converging to zero such that
\begin{equation}
u_t(x) = \lim_{\epsilon_i \rightarrow 0} u_t^{\epsilon_i} (x)
\end{equation}
exists, is locally Lipschitz in $x$ and $t$, and the convergence is locally uniform.  Hence we may define
\begin{equation}
g_t = \lim_{\epsilon_i \rightarrow 0} g_t^{\epsilon_i} = u_t^4(x) g
\end{equation}
for $t \geq 0$.
\end{cor}

Define $\{\widetilde{\Sigma}_\gamma (t)\}$ to be the collections of limit surfaces of $\partial M_{t}^{\epsilon_i}$ in the limit as $\epsilon_i$ approaches $0$.
As discussed in \cite{Bray}, the limiting surfaces $\{\widetilde{\Sigma}_\gamma (t)\}$ are all smooth.

\begin{proposition}
The surface $\widetilde{\Sigma}_{\gamma_2}(t_2)$ encloses $\widetilde{\Sigma}_{\gamma_1} (t_1)$ for all $t_2>t_1\geq 0$ and for any $\gamma_1$ and $\gamma_2$.
\end{proposition}

\begin{proof}
The same arguments as in the proof of Theorem 5 in \cite{Bray} apply here, except for one technical point that needs to be addressed. Namely, in \cite{Bray}, it is used that $v_{t}^{\epsilon}$ is a harmonic function so that the maximum principle applies. In our setting, this function should be replaced by
$u^\epsilon_t v^\epsilon_t$, since here $v^{\epsilon}_{t}$ represents the logarithmic velocity $\frac{d}{dt}u_t^\epsilon = v_{t}^\epsilon u_{t}^\epsilon$, while in \cite{Bray} $v_{t}^\epsilon$ represents the velocity $\frac{d}{dt}u_t^\epsilon=v_t^\epsilon$. Thus it remains to show that $u^\epsilon_t v^\epsilon_t$ satisfies an equation outside $\partial M^\epsilon_t$, to which the maximum principle applies. To see this, note that \eqref{4.6.1} holds with $\epsilon$, and use a standard property for the conformal Laplacian to obtain
\begin{align}
\begin{split}
\Delta_{g}(u_{t}^{\epsilon}v_{t}^{\epsilon})&=(u_{t}^{\epsilon})^{5}\Delta_{g_{t}^{\epsilon}}v_{t}^{\epsilon}
-\frac{1}{8}R_{g_{t}^{\epsilon}}(u_{t}^{\epsilon})^{5}v_{t}^{\epsilon}
+\frac{1}{8}R_{g}u_{t}^{\epsilon}v_{t}^{\epsilon}\\
&=v_{t}^{\epsilon}(u_{t}^{\epsilon})^{5}\left(|E_{t}^{\epsilon}|_{g_{t}^{\epsilon}}^{2}+
|B_{t}^{\epsilon}|_{g_{t}^{\epsilon}}^{2}\right)-\frac{1}{8}R_{g_{t}^{\epsilon}}(u_{t}^{\epsilon})^{5}v_{t}^{\epsilon}
+\frac{1}{8}R_{g}u_{t}^{\epsilon}v_{t}^{\epsilon}\\
&=\left[\frac{3}{4}(u_{t}^{\epsilon})^{4}\left(|E_{t}^{\epsilon}|_{g_{t}^{\epsilon}}^{2}+
|B_{t}^{\epsilon}|_{g_{t}^{\epsilon}}^{2}\right)
+\frac{1}{4}\left(|E|_{g}^{2}+|B|_{g}^{2}\right)\right](u_{t}^{\epsilon}v_{t}^{\epsilon}).
\end{split}
\end{align}
Since the term in brackets on the right-hand side is nonnegative, it follows that the resulting equation for
$u^\epsilon_t v^\epsilon_t$ admits a maximum principle.
\end{proof}

Define $\partial M_{t}$ to be the outermost minimal area enclosure of the original horizon $\partial M$ in $(M, g_t)$. Apart from the proposition above, the rest of the proof of existence of the flow is identical to the arguments in \cite{Bray}. In particular, we have the following result.

\begin{theorem}
The surface $\partial M_{t_2}$ encloses $\partial M_{t_1}$ for all $t_2 > t_1 \geq 0$, and the areas remain constant $|\partial M_{t}|_{g_{t}} = |\partial M|_{g}$ for all $t \geq 0$.
Furthermore, the set $J$ of $t (\geq 0)$ at which point the surface ``jumps", namely when
\begin{equation}
 \lim_{s \rightarrow t^-} \partial M_{s} \neq
\lim_{s \rightarrow t^+} \partial M_{s},
\end{equation}
is countable, and for $t \notin J$, $\widetilde{\Sigma}_\gamma (t)$ is single valued.
Given the horizon $\partial M_{t}$, $v_t$ may be defined as in Section \ref{sec2}, and
serves as the logarithmic velocity of the flow $\frac{d}{dt}u_{t}=v_{t}u_{t}$.
\end{theorem}

\section{Exhaustion}
\label{sec4}

The existence of the charged conformal flow, and its properties listed in Sections \ref{sec2} and \ref{sec3},
are independent of the area/charge inequality $|\partial M|_{g}\geq 4\pi q^{2}$, or equivalently $\rho\geq|q|$ as expressed in the introduction. It is then noteworthy and perhaps surprising, that the property of exhaustion, which states that the flowing surfaces $\partial M_{t}$
eventually enclose any bounded set, essentially holds\footnote{It is proven that the strict area/charge inequality is sufficient for exhaustion, and that the nonstrict area/charge inequality is necessary for exhaustion.}  if and only if the area/charge inequality is valid. In fact, this section is the only place in the paper where the area/charge inequality plays a role. As in \cite{Bray}, the proof will follow two basic steps. The first consists of showing that $\partial M_{t}$ cannot, for all $t\geq 0$, be enclosed by any fixed large coordinate sphere in the asymptotic end, and the second entails showing that it is not possible for $\partial M_{t}$ to be only partially contained, for all $t\geq 0$, in a large coordinate sphere. It turns out that the second step may be proved directly from the same arguments in \cite{Bray}, and does not require the area/charge inequality. Thus, we will focus here on the first step in which the area/charge inequality is needed.

Before proceeding, we show that the area/charge inequality is a necessary condition for exhaustion. Note that if exhaustion occurs, then eventually the surfaces $\partial M_{t}$ become connected.

\begin{lemma}\label{necessary}
If for some $t\geq 0$, $\partial M_{t}$ is connected, then $|\partial M|_{g}\geq 4\pi q^{2}$.
\end{lemma}

\begin{proof}
Since the areas and charges are preserved throughout the flow, it suffices to prove the conclusion at time $t$.
Observe that by the second variation of area formula
\begin{equation}\label{7.1}
0\leq\int_{\partial M_{t}}\left[-\psi\Delta_{\partial M_{t}}\psi-(|\operatorname{II}_{t}|^{2}+\operatorname{Ric}_{g_{t}}(\nu,\nu))\psi^{2}+H_{t}^{2}\psi^{2}\right]dA_{t},\text{ }\text{ }\text{ for any }\text{ }\text{ }\psi\in C^{\infty}(\partial M_{t}),
\end{equation}
where $\operatorname{II}_{t}$ is the second fundamental form and $\operatorname{Ric}_{g_{t}}(\nu,\nu)$ is the Ricci curvature in the normal direction. Since $\partial M_{t}$ is a minimal surface, the Gauss equations yield
\begin{equation}\label{7.2}
|\operatorname{II}_{t}|^{2}+\operatorname{Ric}_{g_{t}}(\nu,\nu)=|\operatorname{II}_{t}|^{2}
+\frac{1}{2}R_{g_{t}}-K_{t}+\frac{1}{2}H_{t}^{2}-\frac{1}{2}|\operatorname{II}_{t}|^{2}
=\frac{1}{2}|\operatorname{II}_{t}|^{2}+\frac{1}{2}R_{g_{t}}-K_{t},
\end{equation}
where $K_{t}$ is Gaussian curvature. It follows that
\begin{equation}\label{7.3}
0\leq\int_{\partial M_{t}}\left(|\nabla\psi|^{2}-\frac{1}{2}|\operatorname{II}_{t}|^{2}\psi^{2}
-\frac{1}{2}R_{g_{t}}\psi^{2}+K_{t}\psi^{2}\right)dA_{t}.
\end{equation}
Choose $\psi\equiv 1$, and note that since $\partial M_{t}$ has spherical topology, the Gauss-Bonnet theorem
and \eqref{2.5} imply that
\begin{align}\label{7.4}
\begin{split}
4\pi \geq\int_{\partial M_{t}}\frac{1}{2}\left(|\operatorname{II}_{t}|^{2}+R_{g_{t}}\right)dA_{t}
&\geq\int_{\partial M_{t}}\left(|E_{t}|_{g_{t}}^{2}+|B_{t}|_{g_{t}}^{2}\right)dA_{t}\\
&\geq\int_{\partial M_{t}}\left(|E_{t}\cdot\nu|^{2}+|B_{t}\cdot\nu|^{2}\right)dA_{t}\\
&\geq|\partial M_{t}|_{g_{t}}^{-1}\left[\left(\int_{\partial M_{t}}E_{t}\cdot\nu dA_{t}\right)^{2}
+\left(\int_{\partial M_{t}}B_{t}\cdot\nu dA_{t}\right)^{2}\right]
=\frac{(4\pi)^{2}q_{t}^{2}}{|\partial M_{t}|_{g_{t}}},
\end{split}
\end{align}
where we have used Jensen's inequality and the fact that the Maxwell fields are divergence free (Theorem \ref{flowproperties}).
\end{proof}

We will now show that the strict area/charge inequality is also a sufficient condition for exhaustion. This will
require some preparation. In \cite{Bray}, it was assumed without loss of generality that the initial data
possessed harmonic asymptotics. Similarly, for the results of this section, we may assume that the initial data $(M,g,E,B)$ possess the so called charged harmonic asymptotics, developed by Corvino in \cite{Corvino}. This means that in the asymptotic end $g=U_{0}^{4}\delta$ (where $\delta$ is the Euclidean metric) for some function $U_{0}$ satisfying $R_{g}=-8U_{0}^{-5}\Delta_{\delta}U_{0}=2|E|_{g}^{2}$, with $E^{i}=U_{0}^{-6}E_{\delta}^{i}$ and $E_{\delta}=-q_{e}\nabla r^{-1}$. Observe that the magnetic field is excluded here, since when the asymptotics are imposed $B$ has the same form as $E$ with $q_{e}$ replaced by $q_{b}$. However, as mentioned in the introduction, nothing is lost by assuming $q_{b}=0$ ($q_{e}=q$), so that $B=0$ in the end with such asymptotics. It should also be noted that the asymptotics used here for the electric field differ slightly from those in \cite{Corvino}, where $E=U_{0}^{-6}\nabla\chi$ for some function $\chi=-qr^{-1}+O(r^{-2})$ which is harmonic in the end; the choice of $\chi$ ensures that $E$ is divergence free on $M$. Thus, in our version of the asymptotics, $E$ is no longer divergence free everywhere, a property which is of no use for the results in the current section.

Write $U_{t}=u_{t}U_{0}$ and $V_{t}=v_{t}u_{t}U_{0}$. Then in the asymptotic end
\begin{equation}\label{7.5}
L_{\delta}U_{t}=U_{t}^{5}L_{g_{t}}1=-\frac{1}{8}U_{t}^{5}R_{g_{t}},
\text{ }\text{ }\text{ }\text{ }\text{ }\text{ }
L_{\delta}V_{t}=U_{t}^{5}L_{g_{t}}v_{t}=U_{t}^{5}\left(|E_{t}|_{g_{t}}^{2}v_{t}-\frac{1}{8}R_{g_{t}}v_{t}\right).
\end{equation}
According to the charged harmonic asymptotics and \eqref{4.6.1} we have that $R_{g_{t}}=2|E_{t}|_{g_{t}}^{2}=2U_{t}^{-8}|E_{\delta}|_{\delta}^{2}$, therefore
\begin{equation}\label{7.7}
\Delta_{\delta}U_{t}=-\frac{1}{4}|E_{\delta}|_{\delta}^{2}U_{t}^{-3},
\text{ }\text{ }\text{ }\text{ }\text{ }\text{ }
\Delta_{\delta}V_{t}=\frac{3}{4}U_{t}^{-4}|E_{\delta}|_{\delta}^{2}V_{t}.
\end{equation}
Let $S_{r(t)}$ be a large coordinate sphere in the asymptotic end, and define $\tilde{V}_{t}$ to be the unique solution of the boundary value problem
\begin{equation}\label{7.9}
\Delta_{\delta}\tilde{V}_{t}=\frac{3}{4}\tilde{U}_{t}^{-4}|E_{\delta}|_{\delta}^{2}\tilde{V}_{t},\text{ }\text{ }\text{ }\tilde{V}_{t}=0\text{ }\text{ }\text{ on }\text{ }\text{ }S_{r(t)},\text{ }\text{ }\text{ }
\tilde{V}_{t}\rightarrow -e^{-t}\text{ }\text{ }\text{ as }\text{ }\text{ }|x|\rightarrow \infty,
\end{equation}
where $\tilde{U}_{t}$ is the function $U_{t}$ in the conformal flow of the Reissner-Nordstr\"{o}m initial data (see Appendix B). Note that $\tilde{V}_{t}$ is the velocity function $V_{t}$ in the conformal flow
of the Reissner-Nordstr\"{o}m initial data, and in particular
\begin{equation}\label{7.10}
\tilde{V}_{t}=\frac{-e^{-2t}+e^{2t}\frac{\tilde{m}^{2}-q^{2}}{4|x|^{2}}}{\sqrt{e^{-2t}+\frac{\tilde{m}}{|x|}+e^{2t}\frac{\tilde{m}^{2}-q^{2}}{4|x|^{2}}}}
\end{equation}
for some constant $\tilde{m}$. We choose $\tilde{m}$ so that the boundary condition of \eqref{7.9} is satisfied, namely
\begin{equation}\label{7.11}
\tilde{m}=\sqrt{4e^{-4t}r(t)^{2}+q^{2}}.
\end{equation}
It follows that
\begin{equation}\label{7.12}
\tilde{U}_{t}=\left(e^{-2t}+\frac{\sqrt{4e^{-4t}r(t)^{2}+q^{2}}}{|x|}+e^{-2t}\frac{r(t)^{2}}{|x|^{2}}\right)^{1/2}.
\end{equation}

For reasons that will become clear in the proof of Proposition \ref{step1} below, we would like to compare the solution of the conformal flow $U_{t}$, or more precisely a radial approximation $\hat{U}_{t}$, with the model solution from the Reissner-Nordstr\"{o}m example $\tilde{U}_{t}$. The desired radial approximation is given
as the unique (radial) solution of
\begin{equation}\label{08.1}
\Delta_{\delta}\hat{U}_{t}=-\frac{1}{4}|E_{\delta}|_{\delta}^{2}\hat{U}_{t}^{-3},\text{ }\text{ }\text{ }\text{ }
\hat{U}_{t}=\left(\frac{1}{4\pi r(t)^{2}}\int_{S_{r(t)}}U_{t}^{4}dA_{\delta}\right)^{1/4}\text{ }\text{ }\text{ on }\text{ }\text{ }S_{r(t)},\text{ }\text{ }\text{ }\text{ }
\hat{U}_{t}\rightarrow e^{-t}\text{ }\text{ }\text{ as }\text{ }\text{ }|x|\rightarrow\infty.
\end{equation}
The corresponding radial velocity function $\hat{V}_{t}=\frac{d}{dt}\hat{U}_{t}$ is the unique solution of the boundary value problem
\begin{equation}\label{08.2}
\Delta_{\delta}\hat{V}_{t}=\frac{3}{4}\hat{U}_{t}^{-3}|E_{\delta}|_{\delta}^{2}\hat{V}_{t},\text{ }\text{ }\text{ }\text{ }\hat{V}_{t}=\frac{d}{dt}[\hat{U}_{t}(r(t))]\text{ }\text{ }\text{ on }\text{ }\text{ }S_{r(t)},
\text{ }\text{ }\text{ }\text{ }\hat{V}_{t}\rightarrow-e^{-t}\text{ }\text{ }\text{ as }\text{ }\text{ }|x|\rightarrow\infty.
\end{equation}
It turns out that $\hat{U}_{t}$ has a relatively simple explicit form.

\begin{lemma}\label{explicit}
Let $r(t)=\varepsilon\sqrt{A_{0}}e^{2t}$ with $A_{0}=|\partial M|_{g}$. If $\varepsilon$ is sufficiently small, then there exists a constant $\alpha>-\frac{1}{2}q^{2}$, depending on $U_{t}|_{S_{r(t)}}$, such that
\begin{equation}\label{08.3}
\hat{U}_{t}^{4}(x)=e^{-4t}+\frac{e^{-2t}\sqrt{\frac{8}{3}\left(\alpha+\frac{1}{2}q^{2}\right)}}{|x|}
+\frac{\alpha}{|x|^{2}}+\frac{e^{2t}\sqrt{\frac{8}{3}\left(\alpha+\frac{1}{2}q^{2}\right)}(\alpha-q^{2})}{6|x|^{3}}
+\frac{e^{4t}(\alpha-q^{2})^{2}}{36|x|^{4}}.
\end{equation}
\end{lemma}

\begin{proof}
Consider the equation satisfied by $\hat{U}_{t}^{4}$:
\begin{equation}\label{08.4}
\Delta_{\delta}\hat{U}_{t}^{4}=-|E_{\delta}|_{\delta}^{2}+\frac{3}{4}\hat{U}_{t}^{-4}|\nabla\hat{U}_{t}^{4}|_{\delta}^{2}.
\end{equation}
Since the equation and all coefficients are analytic in their arguments, we may assume that the solution is given by an expansion
\begin{equation}\label{08.5}
\hat{U}_{t}^{4}=e^{-4t}+\frac{c_{1}}{|x|}+\frac{c_{2}}{|x|^{2}}+\frac{c_{3}}{|x|^{3}}+\cdots.
\end{equation}
We then proceed to calculate each term in \eqref{08.4}. For instance
\begin{equation}\label{08.6}
\Delta_{\delta}\hat{U}_{t}^{4}=\frac{2c_{2}}{|x|^{4}}+\frac{6c_{3}}{|x|^{5}}+\frac{12c_{4}}{|x|^{6}}
+\frac{20c_{5}}{|x|^{7}}+\frac{30c_{6}}{|x|^{8}}+\cdots,
\end{equation}
so that
\begin{align}\label{08.7}
\begin{split}
\hat{U}_{t}^{4}\Delta_{\delta}\hat{U}_{t}^{4}=&\frac{2e^{-4t}c_{2}}{|x|^{4}}+\frac{2c_{1}c_{2}+6e^{-4t}c_{3}}{|x|^{5}}+\frac{2c_{2}^{2}+6c_{1}c_{3}+12e^{-4t}c_{4}}{|x|^{6}}\\
&+\frac{8c_{2}c_{3}+12c_{1}c_{4}+20e^{-4t}c_{5}}{|x|^{7}}+\frac{2c_{2}c_{4}+6c_{3}^{2}+12c_{2}c_{4}+20c_{1}c_{5}+30e^{-4t}c_{6}}{|x|^{8}}+\cdots.
\end{split}
\end{align}
Next observe that
\begin{equation}\label{08.8}
\partial_{r}\hat{U}_{t}^{4}=-\frac{c_{1}}{|x|^{2}}-\frac{2c_{2}}{|x|^{3}}-\frac{3c_{3}}{|x|^{4}}-\frac{4c_{4}}{|x|^{5}}-\frac{5c_{5}}{|x|^{6}}-\cdots,
\end{equation}
which yields
\begin{equation}\label{08.9}
|\nabla\hat{U}_{t}^{4}|^{2}=|\partial_{r}\hat{U}_{t}^{4}|^{2}=\frac{c_{1}^{2}}{|x|^{4}}+\frac{4c_{1}c_{2}}{|x|^{5}}+\frac{4c_{2}^{2}+6c_{1}c_{3}}{|x|^{6}}
+\frac{8c_{1}c_{4}+12c_{2}c_{3}}{|x|^{7}}+\frac{9c_{3}^{2}+16c_{2}c_{4}+10c_{1}c_{5}}{|x|^{8}}+\cdots.
\end{equation}
Finally
\begin{equation}\label{08.10}
\hat{U}_{t}^{4}|E_{\delta}|^{2}=\frac{e^{-4t}q^{2}}{|x|^{4}}+\frac{c_{1}q^{2}}{|x|^{5}}+\frac{c_{2}q^{2}}{|x|^{6}}+\frac{c_{3}q^{2}}{|x|^{7}}+\frac{c_{4}q^{2}}{|x|^{8}}+\cdots.
\end{equation}

By combining these expansions and using equation \eqref{08.4}, we find the following relations
\begin{align}\label{08.11}
\begin{split}
2e^{-4t}c_{2}&=-e^{-4t}q^{2}+\frac{3}{4}c_{1}^{2},\\
6e^{-4t}c_{3}+2c_{1}c_{2}&=-c_{1}q^{2}+3c_{1}c_{2},\\
12e^{-4t}c_{4}+6c_{1}c_{3}+2c_{2}^{2}&=-q^{2}c_{2}+3c_{2}^{2}+\frac{9}{2}c_{1}c_{3},\\
20e^{-4t}c_{5}+12c_{1}c_{4}+6c_{2}c_{3}+2c_{2}c_{3}&=-q^{2}c_{3}+6c_{1}c_{4}+9c_{2}c_{3}.
\end{split}
\end{align}
From this we can solve for the constants $c_{i}$:
\begin{align}\label{08.15}
\begin{split}
c_{2}&=\frac{3}{8}e^{4t}c_{1}^{2}-\frac{1}{2}q^{2},\\
c_{3}&=\frac{1}{16}e^{8t}c_{1}^{3}-\frac{1}{4}e^{4t}q^{2}c_{1},\\
c_{4}&=\frac{3}{48}e^{4t}q^{4}-\frac{1}{32}e^{8t}c_{1}^{2}q^{2}+\frac{1}{256}e^{12t}c_{1}^{4},\\
c_{i}&=0,\text{ }\text{ }\text{ }\text{ }i\geq 5.
\end{split}
\end{align}
Although the higher order terms for $i>5$ have not been computed here, one may deduce that they all vanish by simply checking that \eqref{08.5}, with these coefficients, solves \eqref{08.4}. The constant $c_1$ may be chosen in
order to realize the correct boundary condition.

Let us now obtain the form \eqref{08.3}. The first task is to show that $c_{1}>0$. To see this, first recall the result of Bray and Iga \cite{BrayIga}, which states that
\begin{equation}\label{08.19}
U_{t}^{4}\geq\frac{cA_{0}}{|x|^{2}}\text{ }\text{ }\text{ outside of }\text{ }\text{ }S_{r(t)}
\end{equation}
for some positive constant $c$. Since the average value of $U_{t}^{4}$ agrees with that of $\hat{U}_{t}^{4}$ on $S_{r(t)}$, we find that
\begin{equation}\label{08.20}
\hat{U}_{t}^{4}\geq c\varepsilon^{-2}e^{-4t}>e^{-4t}\text{ }\text{ }\text{ on }\text{ }\text{ }S_{r(t)},
\end{equation}
if $\varepsilon$ is small enough. Moreover
$\hat{U}_{t}^{4}\rightarrow e^{-4t}$ as $|x|\rightarrow\infty$.
From equation \eqref{08.4}, it is clear that $\hat{U}_{t}^{4}$ cannot obtain an interior minimum, thus
\begin{equation}\label{08.22}
\hat{U}_{t}^{4}> e^{-4t}\text{ }\text{ }\text{ outside of }\text{ }\text{ }S_{r(t)}.
\end{equation}
It follows that $c_{1}>0$.
Now solve for $c_{i}$ in terms of $c_{2}$ to produce
\begin{equation}\label{08.23}
c_{1}=e^{-2t}\sqrt{\frac{8}{3}\left(c_{2}+\frac{1}{2}q^{2}\right)},\text{ }\text{ }\text{ }\text{ }
c_{3}=\frac{e^{2t}}{6}\sqrt{\frac{8}{3}\left(c_{2}+\frac{1}{2}q^{2}\right)}(c_{2}-q^{2}),
\text{ }\text{ }\text{ }\text{ }
c_{4}=\frac{1}{36}e^{4t}(c_{2}-q^{2})^{2}.
\end{equation}
The desired result is obtained by setting $\alpha=c_{2}$ and noting that $\alpha>-\frac{1}{2}q^{2}$ by \eqref{08.15}.
\end{proof}

Notice that the Reissner-Nordstr\"{o}m conformal factors have a similar expansion to that of $\hat{U}_{t}$, namely
\begin{align}\label{08.26}
\begin{split}
\tilde{U}_{t}^{4}=&\left(e^{-2t}+\frac{\sqrt{4e^{-4t}r(t)^{2}+q^{2}}}{|x|}+\frac{e^{-2t}r(t)^{2}}{|x|^{2}}\right)^{2}\\
=&e^{-4t}+\frac{2e^{-2t}\sqrt{4e^{-4t}r(t)^{2}+q^{2}}}{|x|}+\frac{6e^{-4t}r(t)^{2}+q^{2}}{|x|^{2}}\\
&+\frac{2e^{-2t}r(t)^{2}\sqrt{4e^{-4t}r(t)^{2}+q^{2}}}{|x|^{3}}+\frac{e^{-4t}r(t)^{4}}{|x|^{4}}.
\end{split}
\end{align}
This is not too surprising, since $\tilde{U}_{t}$ satisfies the same equation \eqref{08.4} as $\hat{U}_{t}^{4}$, and has the same asymptotic behavior as $|x|\rightarrow\infty$. Observe also that
\begin{equation}\label{08.27}
\tilde{V}_{t}=\frac{d}{dt}\tilde{U}_{t}=-\frac{e^{-2t}\left(1-\frac{r(t)^{2}}{|x|^{2}}\right)}{\left(e^{-2t}-\frac{\sqrt{4e^{-4t}r(t)^{2}+q^{2}}}{|x|}+e^{2t}\frac{r(t)^{2}}{|x|^{2}}\right)^{1/2}},
\end{equation}
satisfies
\begin{equation}\label{08.28}
\Delta_{\delta}\tilde{V}_{t}=\frac{3}{4}\tilde{U}_{t}^{-4}|E_{\delta}|_{\delta}^{2}\tilde{V}_{t},\text{ }\text{ }\text{ }\tilde{V}_{t}=0\text{ }\text{ }\text{ on }\text{ }\text{ }S_{r(t)},\text{ }\text{ }\text{ }
\tilde{V}_{t}\rightarrow -e^{-t}\text{ }\text{ }\text{ as }\text{ }\text{ }|x|\rightarrow \infty.
\end{equation}
The next lemma gives the foundational estimate on which the exhaustion proof is based. It is also the
primary place where the area/charge inequality is required.

\begin{lemma}\label{comparison}
Let $r(t)=\varepsilon\sqrt{A_{0}}e^{2t}$ with $A_{0}=|\partial M|_{g}$. If $A_{0}> 4\pi q^{2}$ and $\varepsilon$ is sufficiently small, then
\begin{equation}\label{08.29}
\hat{U}_{t}(x)\geq\tilde{U}_{t}(x)\text{ }\text{ }\text{ for }\text{ }\text{ }|x|\geq r(t).
\end{equation}
\end{lemma}

\begin{proof}
Note that by setting $\alpha=q^{2}+6e^{-4t}r(t)^{2}$, $\hat{U}_{t}$ becomes $\tilde{U}_{t}$. In fact, by directly comparing the coefficients in the expressions for these two functions, it is apparent that the desired result follows if
\begin{equation}\label{abc}
\alpha-q^{2}\geq 6e^{-4t}r(t)^{2}=6\varepsilon^{2}A_{0}\text{ }\text{ }\text{ }\text{ and }\text{ }\text{ }\text{ }(\alpha-q^{2})^{2}\geq 36 e^{-8t}r(t)^{4}= 36\varepsilon^{4}A_{0}^{2}.
\end{equation}

Observe that since $|S_{r(t)}|_{g_{t}}\geq|\partial M_{t}|_{g_{t}}=|\partial M|_{g}$, we have
\begin{align}\label{08.30}
\begin{split}
\frac{A_{0}}{4\pi}&\leq\frac{1}{4\pi}\int_{S_{r(t)}}U_{t}^{4}dA_{\delta}
=\frac{1}{4\pi}\int_{S_{r(t)}}\hat{U}_{t}^{4}dA_{\delta}\\
&=\varepsilon^{2}A_{0}+\varepsilon\sqrt{A_{0}}\sqrt{\frac{8}{3}\left(\alpha+\frac{1}{2}q^{2}\right)}+\alpha
+\frac{1}{6\varepsilon\sqrt{A_{0}}}\sqrt{\frac{8}{3}\left(\alpha+\frac{1}{2}q^{2}\right)}(\alpha-q^{2})
+\frac{(\alpha-q^{2})^{2}}{36\varepsilon^{2}A_{0}}.
\end{split}
\end{align}
Suppose that both inequalities in \eqref{abc} are violated, then
\begin{equation}\label{08.31}
\frac{A_{0}}{4\pi}<2\varepsilon^{2}A_{0}+\varepsilon\sqrt{A_{0}}\sqrt{\frac{8}{3}\left(\frac{3}{2}q^{2}+6\varepsilon^{2}A_{0}\right)}
+q^{2}+6\varepsilon^{2}A_{0}+\varepsilon\sqrt{\frac{8A_{0}}{3}\left(\frac{3}{2}q^{2}+6\varepsilon^{2}A_{0}\right)}.
\end{equation}
However this is impossible for small $\varepsilon$, since $A_{0}> 4\pi q^{2}$ independent of $\varepsilon$. Therefore, at least one of the inequalities in \eqref{abc} must be satisfied. If the first inequality is satisfied, then so is the second. So assume now that the second inequality is satisfied but not the first. The only way that this can happen is if
$\alpha-q^{2}<-6\varepsilon^{2}A_{0}$. We claim however, that $\alpha\geq q^{2}$ as a result of the positive mass theorem with charge \cite{GibbonsHawkingHorowitzPerry}, and hence \eqref{abc} holds.

To verify the claim, consider the initial data $(\mathbb{R}^{3}\setminus \{0\},\hat{U}_{t}^{4}\delta,\hat{U}_{t}^{-6}E_{\delta})$, which satisfies the charged dominant energy condition, the Maxwell constraint, and has mass $\hat{m}=\frac{1}{2}\sqrt{\frac{8}{3}\left(\alpha+\frac{1}{2}q^{2}\right)}$. Note that although $\hat{U}_{t}$ was initially defined in \eqref{08.1} only on $\mathbb{R}^{3}\setminus B_{r(t)}$, the
explicit expression for $\hat{U}_{t}$ in Lemma \ref{explicit} is valid on $\mathbb{R}^{3}\setminus \{0\}$. If $\alpha=q^{2}$ then we are done, so assume that $\alpha\neq q^{2}$. Then according to the expansion of $\hat{U}_{t}$, this initial data set has an asymptotically flat end corresponding to $\{0\}$, and therefore a minimal surface exists which separates the two ends. We may now apply the positive mass theorem with charge to conclude that $\hat{m}\geq|q|$, and the claim follows.
\end{proof}

We have now finished the preparation, and are ready to establish the first step in the proof of exhaustion.

\begin{proposition}\label{step1}
Let $r(t)=\varepsilon\sqrt{A_{0}}e^{2t}$ with $A_{0}=|\partial M|_{g}$. If $A_{0}> 4\pi q^{2}$ and $\varepsilon$ is sufficiently small, then $\partial M_{t}$ cannot be entirely enclosed by the coordinate sphere $S_{r(t)}$ for all $t$.
\end{proposition}

\begin{proof}
The proof is by contradiction. Thus assume that $\partial M_{t}$ is entirely enclosed by $S_{r(t)}$ for all $t\geq\overline{t}$. It will then be shown that if $\overline{t}$ is sufficiently large, then $\partial M_{t}$ is not
the outermost minimal area enclosure of $\partial M$ with respect to the metric $g_{t}$, yielding a contradiction.

Consider the equation satisfied by the difference $W_{t}=\tilde{V}_{t}-\hat{V}_{t}$
\begin{equation}\label{08.32}
\Delta_{\delta}W_{t}=\frac{3}{4}\hat{U}_{t}^{-4}|E_{\delta}|_{\delta}^{2}W_{t}
+\frac{3}{4}(\tilde{U}_{t}^{-4}-\hat{U}_{t}^{-4})\tilde{V}_{t}|E_{\delta}|_{\delta}^{2}.
\end{equation}
Moreover $W_{t}\rightarrow 0$ as $|x|\rightarrow\infty$, and
\begin{equation}\label{08.33}
\hat{V}_{t}(r(t))=\frac{d}{dt}[\hat{U}_{t}(r(t))]<0=\tilde{V}_{t}(r(t))
\end{equation}
so that $W_{t}>0$ on $S_{r(t)}$. In \eqref{08.33} we used the formula \eqref{08.3} to show that $\frac{d}{dt}[\hat{U}_{t}(r(t))]<0$.
Since $\tilde{U}_{t}^{-4}-\hat{U}_{t}^{-4}\geq 0$ by Lemma \ref{comparison}, we may apply the maximum principle to conclude that $W_{t}\geq 0$ outside $S_{r(t)}$.

The remaining arguments proceed similar to those in the proof of Theorem 12 in \cite{Bray}. From the above, we have that $\hat{V}_{t}\leq\tilde{V}_{t}$ outside of $S_{r(t)}$.\footnote{It may also be possible to prove this inequality directly from the explicit formulas for $\hat{V}_{t}$ and $\tilde{V}_{t}$, with help from the area/charge inequality.} This allows an estimate of $\hat{U}_{t}$ from above, since $\hat{V}_{t}=\frac{d}{dt}\hat{U}_{t}$. In this direction, first notice that
\begin{equation}\label{08.34}
\tilde{V}_{t}=-\frac{e^{-2t}\left(1-\frac{r(t)^{2}}{|x|^{2}}\right)}
{\left(e^{-2t}-\frac{\sqrt{4e^{-4t}r(t)^{2}+q^{2}}}{|x|}+e^{2t}\frac{r(t)^{2}}{|x|^{2}}\right)^{1/2}}
\leq-\frac{e^{-2t}\left(1-\frac{r(t)^{2}}{|x|^{2}}\right)}{\left(e^{-2t}+\frac{2e^{-2t}r(t)}{|x|}
+e^{2t}\frac{r(t)^{2}}{|x|^{2}}\right)^{1/2}}
=-e^{-t}\left(1-\frac{r(t)}{|x|}\right).
\end{equation}
Now choose a constant $c>0$ such that $\hat{U}_{\overline{t}}(x)\leq e^{-\overline{t}}+\frac{c}{|x|}$ for all $x$ outside of $S_{r(t)}$.
Then for all $x$ outside of $S_{r(t)}$,
\begin{equation}\label{08.35}
\hat{U}_{t}(x)=\hat{U}_{\overline{t}}(x)+\int_{\overline{t}}^{t}\hat{V}_{s}ds
\leq \hat{U}_{\overline{t}}(x)+\int_{\overline{t}}^{t}\tilde{V}_{s}ds
\leq e^{-t}+\frac{1}{|x|}\left[c+\varepsilon\sqrt{A_{0}}(e^{t}-e^{\overline{t}})\right].
\end{equation}
It follows that
\begin{equation}\label{08.36}
|S_{r(t)}|_{g_{t}}=\int_{S_{r(t)}}U_{t}^{4}dA_{\delta}
=\int_{S_{r(t)}}\hat{U}_{t}^{4}dA_{\delta}
\leq 4\pi r(t)^{2}\hat{U}_{t}^{4}(r(t))
\leq 4\pi\varepsilon^{2}A_{0}[2+O(\varepsilon^{-1}e^{-t})]^{4}.
\end{equation}
Therefore, for $\varepsilon$ sufficiently small and $t\geq\overline{t}$ sufficiently large $|S_{r(t)}|_{g_{t}}<A_{0}$.
\end{proof}

We are now ready to state the main result of this section.

\begin{theorem}\label{exhaustion}
If $|\partial M|_{g}> 4\pi q^{2}$, then the collection of subdomains $\{M_{t}\}$ exhausts the manifold $M$. In particular, the flowing surfaces
$\partial M_{t}$ eventually become connected (topological 2-spheres) for all sufficiently large times.
\end{theorem}

Given Proposition \ref{step1}, the proof of this statement is identical to that which appears in Section 10 of \cite{Bray}, after noting that $U_t$ is superharmonic by \eqref{7.7}.

\section{Monotonicity of the Mass}
\label{sec5}

Monotonicity of the mass is proven with a doubling argument similar to that in \cite{Bray}.
However here, the doubling procedure is based on the proof of uniqueness for the Reissner-N\"{o}rdstrom
black hole given by Masood-ul-Alam \cite{Masood}. Let $(M_{t}^{-}\cup M_{t}^{+},
g_{t}^{\pm})$ be the doubled manifold with $M_{t}^{\pm}$ representing two copies of $M_{t}$ glued along their boundaries and $g_{t}^{\pm}=(w_{t}^{\pm})^{4}g_{t}$,
where
\begin{equation}\label{1.1}
w_{t}^{\pm}=\frac{1}{2}\sqrt{(1\pm \overline{v}_{t})^{2}-\phi_{t}^{2}}.
\end{equation}
The function $\overline{v}_{t}$ (which approximates $v_{t}$) and $\phi_{t}$ imitate the roles
played in the static case by the norm of the Killing field and the electromagnetic potential, respectively. Ultimately though,
these functions are chosen to impart positivity to the scalar curvature of $g_{t}^{\pm}$.
In \cite{Masood}, conformal factors having the same structure as \eqref{1.1}, and built with the aforementioned pieces of static data, were used in the doubling argument. Moreover in this setting of the black hole uniqueness result, the static electrovacuum equations imply nonnegativity of the scalar curvature for the doubled manifold with the aid of a computation similar to that of Lemma \ref{scalar} below.

In order to define $\overline{v}_{t}$, let $\tau_{0}$ be sufficiently small, and set $\tau(x)=\operatorname{dist}_{g_{t}}(x,\partial M_{t})$. Denote surfaces of constant distance to the boundary
and the domain consisting of points whose distance to the boundary is larger than $\tau$,
by $S_{\tau}$ and $M(\tau)$ respectively. Then $\overline{v}_{t}$ is the unique solution of the
boundary value problem
\begin{equation}\label{54}
\Delta_{g_{t}} \overline{v}_{t}-f_{t}\overline{v}_{t}=0\text{ }\text{ }\text{ on }\text{ }\text{ }M_{t},
\text{ }\text{ }\text{ }\text{ }\overline{v}_{t}=0\text{ }\text{ }\text{ on }\text{ }\text{ }\partial M_{t},
\end{equation}
\begin{equation}\label{55}
\overline{v}_{t}=-1+\frac{\overline{\gamma}_{t}}{r}+O\left(\frac{1}{r^{2}}\right)\text{ }\text{ }
\text{ }\text{ as }\text{ }\text{ }\text{ }r\rightarrow\infty,
\end{equation}
where
\begin{equation}\label{56}
f_{t} = \left\{
        \begin{array}{lllll}
            \lambda^{2}\eta(\tau) & \quad \tau<\frac{5}{4}\tau_{0} \\
            |E_{t}|_{g_{t}}^{2}+|B_{t}|_{g_{t}}^{2} & \quad \text{ on } M(2\tau_{0})
        \end{array}
    \right\}.
\end{equation}
Here $\lambda$ is a small parameter to be determined,
$\eta$ is a cut-off function such that $\eta(\tau)=1-2\tau_{0}^{-1}(\tau-\tau_{0})$ for $\frac{3}{4}\tau_{0}<\tau<\frac{5}{4}\tau_{0}$,
$\eta(\tau)=0$ for $\tau<\frac{1}{2}\tau_{0}$,
$|\eta'(\tau)|\leq c\tau_{0}^{-1}$, and
$|\eta''(\tau)|\leq c\tau_{0}^{-2}$.
On the transition region
$\frac{5}{4}\tau_{0}<\tau<2\tau_{0}$, $f_{t}$ is defined so that
\begin{equation}\label{aaaa}
f_{t}\leq \lambda^{2}\eta\left(\frac{5}{4}\tau_{0}\right)+|E_{t}|_{g_{t}}^{2}+|B_{t}|_{g_{t}}^{2}
=\frac{1}{2}\lambda^{2}+|E_{t}|_{g_{t}}^{2}+|B_{t}|_{g_{t}}^{2},
\end{equation}
and so as to make a smooth positive function on $M(\tau_{0})$.


The function $\phi_{t}$ is also defined piecewise. Namely it
will be shown in the next section that if $\lambda,\tau_{0}>0$ and $\tau_{0}$ is sufficiently small, then there is a positive solution of the following
Dirichlet problem
\begin{equation}\label{1.2}
\Delta_{g_{t}}\phi_{t}-\frac{\nabla \overline{v}_{t}\cdot\nabla\phi_{t}}{\overline{v}_{t}}
=\frac{\Lambda}{\phi_{t}}\left|f_{t}-\frac{|\nabla\phi_{t}|^{2}}
{\overline{v}_{t}^{2}}\right|\text{ }\text{ }\text{ on }\text{ }\text{ }M(\tau_{0}),
\end{equation}
\begin{equation}\label{1.3}
\phi_{t}=\lambda\tau_{0}^{4}\text{ }\text{ }\text{ on }\text{ }\text{ }S_{\tau_{0}},\text{ }\text{ }\text{ }\text{ }\text{ }
\phi_{t}\rightarrow 0\text{ }\text{ }\text{ as }\text{ }\text{ }r\rightarrow\infty,
\end{equation}
where $\Lambda$ is a positive constants to be
specified. On the interior region define
\begin{equation}\label{1.5'}
\phi_{t}=\lambda\tau_{0}^{4}+\partial_{\tau}\phi_{t}|_{S_{\tau_{0}}}
(\tau-\tau_{0})\eta(\tau)\text{ }\text{ }\text{ when }\text{ }\text{ }
0\leq\tau\leq\tau_{0},
\end{equation}
where there is a slight abuse of notation in that $\partial_{\tau}\phi_{t}|_{S_{\tau_{0}}}$ is defined for
all $0\leq \tau\leq\tau_{0}$ by the fact that it is constant along the geodesic flow.
Observe that $\phi_{t}$ is $C^{1,1}$ across $S_{\tau_{0}}$, and thus
$\phi_{t}$ is $C^{1,1}(M_{t})$.

The constructions above cut-off contributions from the electromagnetic energy density near the horizon. The purpose of this modification is to avoid possible singular behavior in the pseudo-potential $\phi_{t}$ at the horizon. Heuristically, this does not affect the main arguments used to establish monotonicity of the mass, since the first variation \eqref{mass} of the mass depends on the monopoles $\gamma_{t}$ and $\overline{\gamma}_{t}$, which are arbitrarily close (Theorem \ref{thm15}) as a result of the cut-off taking place on a sufficiently small region.
Notice also that if $|E|_{g}=|B|_{g}=0$ and $\lambda=0$ then $f_{t}=0$, which implies that $\phi_{t}=0$ and $\overline{v}_{t}=v_{t}$. It follows that in the absence of the electromagnetic field, the conformal factors \eqref{1.1} reduce, modulo the choice of $\lambda$, to the same expressions used in \cite{Bray}. Let us now establish positivity of the conformal factors.

\begin{lemma}\label{lemmapos}
If $\lambda$, $\tau_{0}$ are appropriately small, and $\Lambda>1$, then $(1\pm \overline{v}_{t})^{2}-\phi_{t}^{2}>0$ on $M_{t}$.
\end{lemma}

\begin{proof}
Observe that $(1\pm \overline{v}_{t})^{2}-\phi_{t}^{2}=(1\pm \overline{v}_{t}+\phi_{t})(1\pm \overline{v}_{t}-\phi_{t})$. Thus since $\overline{v}_{t}<0$ and $\phi_{t}>0$, it is enough to show that $1+\overline{v}_{t}-\phi_{t}>0$. First we show this on $M(2\tau_{0})$. Equations \eqref{54} and \eqref{1.2} imply that
\begin{align}\label{267}
\begin{split}
\Delta_{g_{t}}(1+\overline{v}_{t}-\phi_{t})
+\frac{\nabla\phi_{t}\cdot\nabla(1+\overline{v}_{t}-\phi_{t})}{\overline{v}_{t}}
&=\left(|E_{t}|_{g_{t}}^{2}+|B_{t}|_{g_{t}}^{2}
-\frac{|\nabla\phi_{t}|^{2}}{\overline{v}_{t}^{2}}\right)\overline{v}_{t}
-\frac{\Lambda}{\phi_{t}}\left|f_{t}-\frac{|\nabla\phi_{t}|^{2}}{\overline{v}_{t}^{2}}\right|\\
&\leq -\left(\frac{\Lambda}{\phi_{t}}+\overline{v}_{t}\right)\left|f_{t}
-\frac{|\nabla\phi_{t}|^{2}}{\overline{v}_{t}^{2}}\right|
+(|E_{t}|_{g_{t}}^{2}+|B_{t}|_{g_{t}}^{2}-f_{t})\overline{v}_{t}\\
&\leq -\left(\frac{\Lambda}{\phi_{t}}+\overline{v}_{t}\right)\left|f_{t}
-\frac{|\nabla\phi_{t}|^{2}}{\overline{v}_{t}^{2}}\right|.
\end{split}
\end{align}
Also
\begin{equation}\label{268}
(1+\overline{v}_{t}-\phi_{t})|_{S_{2\tau_{0}}}>0,\text{ }\text{ }\text{ }\text{ }\text{ }
(1+\overline{v}_{t}-\phi_{t})\rightarrow 0\text{ }\text{ }\text{ as }\text{ }\text{ }r\rightarrow\infty,
\end{equation}
if $\lambda$, $\tau_{0}$ are small enough. Clearly the right-hand side of \eqref{267} is nonpositive if $\Lambda>1$. This is due to the fact that $\phi_{t}\leq\lambda\tau_{0}^{4}$ and $|\overline{v}_{t}|<1$ by the maximum principle. Thus, by the minimum principle
$1+\overline{v}_{t}-\phi_{t}>0$ on $M(2\tau_{0})$.

Let us now consider the remaining region. Again use the fact that $\phi_{t}\leq\lambda\tau_{0}^{4}$ on the region between $S_{\tau_{0}}$ and $S_{2\tau_{0}}$. Moreover, in Appendix A it is shown that $\phi_{t}\leq c(\tau_{0})\lambda$ for $\tau<\tau_{0}$. Hence, since $|\overline{v}_{t}|\leq c\tau_{0}$ on $M_{t}\setminus M(2\tau_{0})$ (by the mean value theorem), if $\lambda$, $\tau_{0}$ are chosen appropriately (small), then $1+\overline{v}_{t}-\phi_{t}>0$ for $\tau<2\tau_{0}$.
\end{proof}

In order to justify the use of $\overline{v}_{t}$ in place of $v_{t}$, in connection with monotonicity of the mass, it must be established that the monopoles of these two functions at spatial infinity remain arbitrarily close.

\begin{theorem}\label{thm15}
If an upper bound for $\lambda$ is fixed, then $|\gamma_{t}-\overline{\gamma}_{t}|\leq c\tau_{0}^{1/4}$ where $c$ is independent of $\lambda$ and $\tau_{0}$.
\end{theorem}

\begin{proof}
We have
\begin{equation}\label{247}
\Delta_{g_{t}}(v_{t}- \overline{v}_{t})-f_{t}(v_{t}-\overline{v}_{t})
= (|E_{t}|_{g_{t}}^{2}+|B_{t}|_{g_{t}}^{2}-f_{t})v_{t}
\text{ }\text{ }\text{ }\text{ on }\text{ }\text{ }\text{ }M_{t},
\end{equation}
\begin{equation}\label{248}
(v_{t}-\overline{v}_{t})|_{\partial M_{t}}=0,\text{ }\text{ }\text{ }\text{ }
(v_{t}-\overline{v}_{t})=\frac{\gamma_{t}-\overline{\gamma}_{t}}{r}
+O\left(\frac{1}{r^{2}}\right)\text{ }\text{ }
\text{ as }\text{ }\text{ }r\rightarrow\infty.
\end{equation}
Thus
\begin{align}\label{249}
\begin{split}
-4\pi(\gamma_{t}-\overline{\gamma}_{t})&=
\int_{M_{t}}\Delta_{g_{t}}(v_{t}-\overline{v}_{t})
+\int_{\partial M_{t}}\partial_{\tau}(v_{t}-\overline{v}_{t})\\
&=\int_{M_{t}}[f_{t}(v_{t}-\overline{v}_{t})+(|E_{t}|_{g_{t}}^{2}+|B_{t}|_{g_{t}}^{2}-f_{t})v_{t}]
+\int_{\partial M_{t}}\partial_{\tau}(v_{t}-\overline{v}_{t}),
\end{split}
\end{align}
and hence
\begin{equation}\label{250}
|\gamma_{t}-\overline{\gamma}_{t}|\leq
C\left(|v_{t}-\overline{v}_{t}|_{C^{1}(\partial M_{t})}
+\parallel f_{t}(v_{t}-\overline{v}_{t})\parallel_{L^{1}(M_{t})}
+\parallel|E_{t}|_{g_{t}}^{2}+|B_{t}|_{g_{t}}^{2}-f_{t}\parallel_{L^{1}(M_{t})}\right).
\end{equation}

In order to estimate $|v_{t}-\overline{v}_{t}|_{C^{1}(\partial M_{t})}$, use that
\begin{equation}\label{251}
|v_{t}-\overline{v}_{t}|_{C^{1}(\partial M_{t})}
\leq|v_{t}-\overline{v}_{t}|_{C^{1}(M_{t})}\leq C\parallel v_{t}-\overline{v}_{t}\parallel_{W^{2,p}(M_{t})}
\end{equation}
for $p>3$. By the $L^{p}$ estimates for \eqref{247}
\begin{equation}\label{252}
\parallel v_{t}-\overline{v}_{t}\parallel_{W^{2,p}(M_{t})}
\leq C(\parallel |E_{t}|_{g_{t}}^{2}+|B_{t}|_{g_{t}}^{2}-f_{t}\parallel_{L^{p}(M_{t})}
+\parallel v_{t}-\overline{v}_{t}\parallel_{L^{p}(M_{t})}).
\end{equation}
We choose $p$ large enough and even, and estimate $\parallel v_{t}-\overline{v}_{t}\parallel_{L^{p}(M_{t})}$.
To do this, multiply the equation \eqref{247} by $(v_{t}-\overline{v}_{t})^{\frac{p}{3}-1}$ and integrate
by parts. It follows that
\begin{equation}\label{253}
\int_{M_{t}}\left(\frac{p}{3}-1\right)\left(\frac{6}{p}\right)^{2}
|\nabla(v_{t}-\overline{v}_{t})^{\frac{p}{6}}|^{2}
+f_{t}|v_{t}-\overline{v}_{t}|^{\frac{p}{3}}
=-\int_{M_{t}}(v_{t}-\overline{v}_{t})^{\frac{p}{3}-1}h_{t},
\end{equation}
where $h_{t}=v_{t}(|E_{t}|_{g_{t}}^{2}+|B_{t}|_{g_{t}}^{2}-f_{t})$. Since $(v_{t}-\overline{v}_{t})^{\frac{p}{6}}=0$
on $\partial M_{t}$ and vanishes sufficiently fast as $r\rightarrow\infty$,
we may apply the Gagliardo-Nirenberg-Sobolev inequality to obtain
\begin{align}\label{254}
\begin{split}
\int_{M_{t}}|v_{t}-\overline{v}_{t}|^{p}&=\int_{M_{t}}\left(|v_{t}-\overline{v}_{t}|^{\frac{p}{6}}\right)^{6}\\
&\leq C\left(\int_{M_{t}}|\nabla(v_{t}-\overline{v}_{t})^{\frac{p}{6}}|^{2}\right)^{3}\\
&\leq C\left(\int_{M_{t}}|v_{t}-\overline{v}_{t}|^{\frac{p}{3}-1}|h_{t}|\right)^{3}\\
&\leq C\left(\int_{M_{t}}\left(|v_{t}-\overline{v}_{t}|^{\frac{p}{3}-1}\right)^{\overline{q}}\right)^{\frac{3}{\overline{q}}}
\left(\int_{M_{t}}|h_{t}|^{\overline{p}}\right)^{\frac{3}{\overline{p}}},
\end{split}
\end{align}
where $\overline{p}^{-1}+\overline{q}^{-1}=1$. We want
$\overline{q}(\frac{p}{3}-1)=p$, which implies
$\overline{q}=\frac{3p}{p-3}$ and $\overline{p}=\frac{3p}{2p+3}$. Thus
\begin{equation}\label{257}
\parallel v_{t}-\overline{v}_{t}\parallel_{L^{p}(M_{t})}\leq C\parallel h_{t}\parallel_{L^{\frac{3p}{2p+3}}(M_{t})}
\leq C\parallel|E_{t}|_{g_{t}}^{2}+|B_{t}|_{g_{t}}^{2}-f_{t}\parallel_{L^{\frac{3p}{2p+3}}(M_{t})}.
\end{equation}
It follows that
\begin{equation}\label{258}
\parallel v_{t}-\overline{v}_{t}\parallel_{W^{2,p}(M_{t})}\leq
C\left(\parallel|E_{t}|_{g_{t}}^{2}+|B_{t}|_{g_{t}}^{2}-f_{t}\parallel_{L^{\frac{3p}{2p+3}}(M_{t})}
+\parallel|E_{t}|_{g_{t}}^{2}+|B_{t}|_{g_{t}}^{2}-f_{t}\parallel_{L^{p}(M_{t})}\right).
\end{equation}

Notice also that if $B(r)$ denotes the domain contained within the coordinate sphere $S_{r}$ in the asymptotic
end, then for $r_{0}$ sufficiently large
\begin{equation}\label{259}
\parallel f_{t}(v_{t}-\overline{v}_{t})\parallel_{L^{1}(M_{t})}\leq
C\left(\int_{B(r_{0})}|v_{t}-\overline{v}_{t}|
+\int_{M_{t}\setminus B(r_{0})}r^{-4}|v_{t}-\overline{v}_{t}|\right),
\end{equation}
since $f_{t}\leq cr^{-4}$ on $M_{t}\setminus B(r_{0})$. Thus
\begin{equation}\label{260}
\parallel f_{t}(v_{t}-\overline{v}_{t})\parallel_{L^{1}(M_{t})}\leq
C\left( \vol(B(r_{0}))^{\frac{1}{q}}\parallel v_{t}-\overline{v}_{t}\parallel_{L^{p}(B(r_{0}))}
+\parallel r^{-4}\parallel_{L^{q}(M_{t}\setminus B(r_{0}))}\parallel v_{t}-\overline{v}_{t}\parallel_{L^{p}(M_{t}\setminus B(r_{0}))}\right),
\end{equation}
so that
\begin{equation}\label{261}
\parallel f_{t}(v_{t}-\overline{v}_{t})\parallel_{L^{1}(M_{t})}\leq C\parallel v_{t}-\overline{v}_{t}\parallel_{L^{p}(M_{t})}.
\end{equation}
Therefore
\begin{equation}\label{262}
|\gamma_{t}-\overline{\gamma}_{t}|\leq
C\left(\parallel|E_{t}|_{g_{t}}^{2}+|B_{t}|_{g_{t}}^{2}-f_{t}\parallel_{L^{\frac{3p}{2p+3}}(M_{t})}
+\parallel|E_{t}|_{g_{t}}^{2}+|B_{t}|_{g_{t}}^{2}-f_{t}\parallel_{L^{p}(M_{t})}\right),
\end{equation}
for $p>3$.

Recall the definition of $f_{t}$ in \eqref{56}. We see that
$|E_{t}|_{g_{t}}^{2}+|B_{t}|_{g_{t}}^{2}-f_{t}\equiv 0$ except on a set of small measure depending on $\tau_{0}$. In particular by choosing $p=4$ we obtain the desired result.
\end{proof}

As in Bray's doubling argument \cite{Bray}, we have that the time derivative of the
mass of $g_{t}$ is given by
\begin{equation}\label{mass}
m_{t}'=-2(m_{t}-e^{-2t}\gamma_{t})=-2\widetilde{m}_{t}+2e^{-2t}(\gamma_{t}-\overline{\gamma}_{t}),
\end{equation}
where $\widetilde{m}_{t}=m_{t}-e^{-2t}\overline{\gamma}_{t}$ is the mass of the doubled manifold. Note that $\phi_{t}$ can be ignored in this computation, since $\phi_{t}=O(r^{-1})$ as $r\rightarrow\infty$. Moreover, since $\partial_{\tau}\phi_{t}=0$ at $\partial M_{t}$, the mean curvatures across the glued boundaries agree. Therefore as $\phi_{t}\in C^{1,1}(M_{t})$ and $\overline{v}_{t}\in C^{\infty}(M_{t})$, we may apply the positive mass theorem with corners \cite{Miao}, \cite{ShiTam} to conclude that $\widetilde{m}_{t}\geq 0$, provided the scalar curvature of the doubled manifold is nonnegative. Thus, it remains to show that the scalar curvature is nonnegative.


\begin{lemma}\label{scalar}
The scalar curvature of the doubled manifold is given by
\begin{align}\label{146}
\begin{split}
R_{g^{\pm}_{t}}=&\frac{1}{2}(\overline{v}_{t})^{-2}(w_{t}^{\pm})^{-8}\left|\phi_{t} \overline{v}_{t}\nabla \overline{v}_{t}-\frac{1}{2}(\overline{v}_{t}^{2}+\phi_{t}^{2}-1)\nabla\phi_{t}\right|^{2}\\
&+(w_{t}^{\pm})^{-4}(R_{g_{t}}-2|E_{t}|_{g_{t}}^{2}-2|B_{t}|_{g_{t}}^{2})
+2(w^{\pm}_{t})^{-4}(|E_{t}|_{g_{t}}^{2}+|B_{t}|_{g_{t}}^{2}-f_{t})\\
&+\frac{1}{2}(w_{t}^{\pm})^{-6}\left[((1\pm \overline{v}_{t})^{2}-\phi_{t}^{2}\mp4\overline{v}_{t}(1\pm \overline{v}_{t}))
\left(f_{t}-\frac{|\nabla\phi_{t}|^{2}}{\overline{v}_{t}^{2}}\right)
+\phi_{t}\left(\Delta_{g_{t}}\phi_{t}-\frac{\nabla \overline{v}_{t}\cdot\nabla\phi_{t}}{\overline{v}_{t}}\right)\right].
\end{split}
\end{align}
\end{lemma}

\begin{proof}
A standard formula yields
\begin{equation}\label{6.12}
R_{g_{t}^{-}}=-8(w_{t}^{-})^{-5}\left(\Delta_{g_{t}}w_{t}^{-}-\frac{1}{8}R_{g_{t}}w_{t}^{-}\right).
\end{equation}
Next, observe that
\begin{equation}\label{6.13}
\nabla w_{t}^{-}=-\frac{1}{2}\frac{(1-\overline{v}_{t})\nabla \overline{v}_{t}+\phi_{t}\nabla\phi_{t}}{\sqrt{(1-\overline{v}_{t})^{2}-\phi_{t}^{2}}}
\end{equation}
so
\begin{equation}\label{6.14}
\Delta_{g_{t}}w_{t}^{-}=-\frac{1}{2}\frac{(1-\overline{v}_{t})\Delta_{g_{t}}\overline{v}_{t}
+\phi_{t}\Delta_{g_{t}}\phi_{t}}{\sqrt{(1-\overline{v}_{t})^{2}-\phi_{t}^{2}}}+\frac{1}{2}\frac{|\nabla \overline{v}_{t}|^{2}-|\nabla\phi_{t}|^{2}}{\sqrt{(1-\overline{v}_{t})^{2}-\phi_{t}^{2}}}
-\frac{1}{2}\frac{|(1-\overline{v}_{t})\nabla \overline{v}_{t}+\phi_{t}\nabla\phi_{t}|^{2}}{((1-\overline{v}_{t})^{2}-\phi_{t}^{2})^{3/2}}.
\end{equation}
It follows that
\begin{align}\label{145}
\begin{split}
R_{g_{t}^{-}}=&\frac{1}{2}(\overline{v}_{t})^{-2}(w_{t}^{-})^{-8}\left|\phi_{t} \overline{v}_{t}\nabla \overline{v}_{t}-\frac{1}{2}(\overline{v}_{t}^{2}+\phi_{t}^{2}-1)\nabla\phi_{t}\right|^{2}
+(w_{t}^{-})^{-4}(R_{g_{t}}-2|E_{t}|_{g_{t}}^{2}-2|B_{t}|_{g_{t}}^{2})\\
&
+2(w_{t}^{-})^{-4}(|E_{t}|_{g_{t}}^{2}+|B_{t}|_{g_{t}}^{2}-f_{t})+
2(w_{t}^{-})^{-6}(1- \overline{v}_{t})(\Delta_{g_{t}} \overline{v}_{t}-f_{t}\overline{v}_{t})\\
&+\frac{1}{2}(w_{t}^{-})^{-6}\left[((1- \overline{v}_{t})^{2}-\phi_{t}^{2}+4\overline{v}_{t}(1- \overline{v}_{t}))
\left(f_{t}-\frac{|\nabla\phi_{t}|^{2}}{\overline{v}_{t}^{2}}\right)
+\phi_{t}\left(\Delta_{g_{t}}\phi_{t}-\frac{\nabla \overline{v}_{t}\cdot\nabla\phi_{t}}{\overline{v}_{t}}\right)\right].
\end{split}
\end{align}
Since $\overline{v}_{t}$ satisfies \eqref{54}, we obtain the desired result for $R_{g_{t}^{-}}$. A similar calculation yields the formula for $R_{g_{t}^{+}}$.
\end{proof}

We are now ready to establish monotonicity of the mass.

\begin{theorem}\label{monotonicity}
The mass $m_{t}$ is nonincreasing.
\end{theorem}

\begin{proof}
First, we note that there exist small perturbations $g_{t}\rightarrow g_{t}^{\varepsilon}$, $E_{t}\rightarrow E_{t}^{\varepsilon}$, $B_{t}\rightarrow B_{t}^{\varepsilon}$
for which a strict charged dominant energy condition holds on $M_{t}$, and in particular
\begin{equation}\label{strictdec}
R_{g_{t}^{\varepsilon}}-2\left(|E_{t}^{\varepsilon}|_{g_{t}^{\varepsilon}}^{2}
+|B_{t}^{\varepsilon}|_{g_{t}^{\varepsilon}}^{2}\right)\geq C_{t}^{\varepsilon}\text{ }\text{ }\text{ }\text{ on }\text{ }\text{ }\text{ }M_{t}\setminus M(\tau_{0}),
\end{equation}
for some constant $C_{t}^{\varepsilon}>0$. Here $\varepsilon$ is the perturbation parameter, and $C_{t}^{\varepsilon}\rightarrow 0$ as $\varepsilon\rightarrow 0$. Moreover, this perturbation may be constructed to preserve the property of minimality for the boundary, as well as the divergence free property of the Maxwell fields. After the perturbation, although $\partial M_{t}$ is minimal it may not be outermost, however, for the present purpose the outermost condition is not necessary. That is, the doubling argument only requires that the boundary be minimal, not outermost. Let us now construct the desired deformation. Fix a smooth positive function $\varrho$ which vanishes sufficiently fast at spatial infinity, and solve the semi-linear boundary value problem
\begin{equation}\label{123456}
\Delta_{g_{t}}z_{t,\varepsilon}-\frac{1}{8}R_{g_{t}}z_{t,\varepsilon}
+\frac{1}{8}\left(R_{g_{t}}+\varepsilon\varrho\right)z_{t,\varepsilon}^{-3}=0\text{ }\text{ }\text{ }\text{ on }\text{ }\text{ }\text{ }M_{t},
\end{equation}
\begin{equation}
\partial_{\tau}z_{t,\varepsilon}=0\text{ }\text{ }\text{ on }\text{ }\text{ }\partial M_{t},\text{ }\text{ }\text{ }\text{ }\text{ }z_{t,\varepsilon}\rightarrow 1\text{ }\text{ }\text{ as }\text{ }\text{ }r\rightarrow\infty.
\end{equation}
It is easily seen that a smooth positive solution exists for small $\varepsilon$, and is pointwise close to 1, by the implicit function theorem. Furthermore, equation \eqref{123456} implies that the conformal metric
$g_{t}^{\varepsilon}=z_{t,\varepsilon}^{4}g_{t}$ has scalar curvature $R_{g_{t}^{\varepsilon}}=
\left(R_{g_{t}}+\varepsilon\varrho\right)z_{t,\varepsilon}^{-8}$, and the Neumann boundary condition guarantees
that the boundary is still a minimal surface with respect to the new metric. Now define
$(E_{t}^{\varepsilon})^{i}=z_{t,\varepsilon}^{-6}E_{t}^{i}$ and $(B_{t}^{\varepsilon})^{i}=z_{t,\varepsilon}^{-6}B_{t}^{i}$, so that the perturbed Maxwell fields remain divergence free. Lastly, a strict charged dominant energy condition holds
\begin{equation}\label{bbbb}
R_{g_{t}^{\varepsilon}}=\left(R_{g_{t}}+\varepsilon\varrho\right)z_{t,\varepsilon}^{-8}
\geq 2z_{t,\varepsilon}^{-8}\left(|E_{t}|_{g_{t}}^{2}+|B_{t}|_{g_{t}}^{2}\right)+\frac{\varepsilon}{2}\varrho
>2\left(|E_{t}^{\varepsilon}|_{g_{t}^{\varepsilon}}^{2}+|B_{t}^{\varepsilon}|_{g_{t}^{\varepsilon}}^{2}\right).
\end{equation}
The constant $C_{t}^{\varepsilon}$ may then be taken to be $\min_{M_{t}\setminus M(\tau_{0})}\frac{\varepsilon}{2}\varrho$.

We will now apply the doubling argument to $(M_{t},g_{t}^{\varepsilon},E_{t}^{\varepsilon},B_{t}^{\varepsilon})$. Note that since $\lambda,\tau_{0}>0$ are small enough, Theorem \ref{existence1} guarantees existence of the conformal factors \eqref{1.1}. According to \eqref{mass}
\begin{equation}\label{mass1}
m_{t}'=-2\widetilde{m}_{t}^{\varepsilon}+2e^{-2t}(\gamma_{t}-\overline{\gamma}_{t})+\theta_{t}^{\varepsilon},
\end{equation}
where $\theta_{t}^{\varepsilon}\rightarrow 0$ as $\varepsilon\rightarrow 0$.
Since $|\gamma_{t}-\overline{\gamma}_{t}|$ and $\theta_{t}^{\varepsilon}$ may be made arbitrarily small, it will follow that $m_{t}'\leq 0$ if the mass
of the doubled manifold is nonnegative $\widetilde{m}_{t}^{\varepsilon}\geq 0$.
In light of the discussion preceding Lemma \ref{scalar}, it suffices to show that the scalar curvature of the doubled manifold is nonnegative. This will be accomplished in two cases associated with different regions. For convenience, in what follows, the superscript $\varepsilon$ will be omitted from most of the notation. \bigskip

\noindent\textit{Case 1: $\tau\geq\tau_{0}$.}\bigskip

In this region, with the help of \eqref{bbbb}, we find that
\begin{align}\label{147}
\begin{split}
R_{g_{t}^{\pm}}\geq&\frac{1}{2}(w_{t}^{\pm})^{-6}\left[((1\pm \overline{v}_{t})^{2}-\phi_{t}^{2}
\mp4\overline{v}_{t}(1\pm\overline{v}_{t}))
\left(f_{t}-\frac{|\nabla\phi_{t}|^{2}}{\overline{v}_{t}^{2}}\right)
+\Lambda\left|f_{t}-\frac{|\nabla\phi_{t}|^{2}}{\overline{v}_{t}^{2}}\right|\right]\\
&+(w_{t}^{\pm})^{-4}\left[\frac{\varepsilon}{2}\varrho+2\left(|E_{t}|_{g_{t}}^{2}+|B_{t}|_{g_{t}}^{2}-f_{t}\right)\right].
\end{split}
\end{align}
The first line on the right-hand side is clearly nonnegative if $\Lambda\geq 12$, since $|\overline{v}_{t}|<1$. Moreover, the second line is nonnegative for $\tau\geq2\tau_{0}$ in light of \eqref{56}, and is nonnegative for
$\tau_{0}\leq\tau\leq 2\tau_{0}$ by \eqref{aaaa} if $\lambda$ is sufficiently small, depending on $\varepsilon$.\bigskip

\noindent\textit{Case 2: $0\leq\tau<\tau_{0}$.}
\bigskip

In this region \eqref{strictdec} holds, and $|\overline{v}_{t}|,\phi_{t},f_{t}\sim 0$, $w_{t}^{\pm}\sim 1$. Therefore \eqref{146} implies
\begin{equation}\label{216}
R_{g_{t}^{\pm}}
\geq\frac{1}{2}(w_{t}^{\pm})^{-6}\left[C_{t}^{\varepsilon}-4f_{t}
-2\frac{|\nabla\phi_{t}|^{2}}{\overline{v}_{t}^{2}}
+\phi_{t}\left(\Delta_{g_{t}}\phi_{t}-\frac{\nabla \overline{v}_{t}\cdot\nabla\phi_{t}}{\overline{v}_{t}}\right)
\right].
\end{equation}
According to \eqref{1.5'} we may write
\begin{equation}\label{216.1}
\phi_{t}=\lambda\tau_{0}^{4}+\beta(\overline{x},\tau)
(\tau-\tau_{0})\eta(\tau)\text{ }\text{ }\text{ when }\text{ }\text{ }
0\leq\tau<\tau_{0},
\end{equation}
where $\overline{x}$ are coordinates on $S_{\tau}$ and $\beta(\overline{x},\tau)=\partial_{\tau}\phi_{t}(\overline{x},\tau_{0})$. Then
\begin{equation}\label{216.2}
|\partial_{\tau}\phi_{t}|\leq|\beta|,\text{ }\text{ }\text{ }\text{ } |\partial_{\tau}^{2}\phi_{t}|\leq c\tau_{0}^{-1}|\beta|,\text{ }\text{ }\text{ }\text{ } |\overline{\nabla}\phi_{t}|
\leq c\tau_{0}|\overline{\nabla}\beta|,\text{ }\text{ }\text{ }\text{ } |\overline{\nabla}^{2}\phi_{t}|
\leq c\tau_{0}|\overline{\nabla}^{2}\beta|,
\end{equation}
where $\overline{\nabla}$ represents the induced connection on $S_{\tau}$.
Estimates for $\beta$ are established in Appendix A (Theorem \ref{thm14}), namely
\begin{equation}\label{216.3}
|\beta|+|\overline{\nabla}\beta|+|\overline{\nabla}^{2}\beta|\leq c(\tau_{0},\varepsilon)\lambda.
\end{equation}
Here, unlike in the appendix, the constant $c(\tau_{0},\varepsilon)$ depends on $\varepsilon$ since $\phi_{t}$ depends on the perturbed initial data.
It follows that
\begin{align}\label{217}
\begin{split}
\phi_{t}\left|\Delta_{g_{t}}\phi_{t}-\frac{\nabla\overline{v}_{t}\cdot\nabla\phi_{t}}{\overline{v}_{t}}\right|
&=\phi_{t}\left|\overline{\Delta}\phi_{t}+\partial_{\tau}^{2}\phi_{t}+H\partial_{\tau}\phi_{t}
-\frac{\nabla \overline{v}_{t}\cdot\nabla\phi_{t}}{\overline{v}_{t}}\right|\\
&\leq
c\phi_{t}\left(\tau_{0}^{-1}|\beta|+\tau_{0}(|\overline{\nabla}\beta|+|\overline{\nabla}^{2}\beta|)\right)\\
&\leq c(\tau_{0},\varepsilon)\lambda^{2}.
\end{split}
\end{align}
Similarly
\begin{equation}\label{217.1}
\frac{|\nabla\phi_{t}|^{2}}{\overline{v}_{t}^{2}}\leq c\left(\frac{|\partial_{\tau}\phi_{t}|^{2}
+|\overline{\nabla}\phi_{t}|^{2}}{\tau_{0}^{2}}\right)\leq c\left(\tau_{0}^{-2}|\beta|^{2}+|\overline{\nabla}\beta|^{2}\right)
\leq c(\tau_{0},\varepsilon)\lambda^{2}.
\end{equation}
Therefore by choosing $\lambda$ sufficiently small, dependent on $\varepsilon$ and $\tau_{0}$, we find that the scalar curvature is nonnegative.\bigskip
\end{proof}

\section{Existence of the Conformal Factor}
\label{sec6}

In this section we will show that a positive solution of the Dirichlet problem \eqref{1.2}, \eqref{1.3} exists, by constructing solutions to an auxiliary problem on a finite domain and then taking a limit as the finite domains exhaust $M(\tau_{0})$.

Let $M(\tau_{0},r_{0})$ denote the complement in $M_{t}$ of the region of distance less than $\tau_{0}$ from $\partial M_{t}$ and the region outside $S_{r_{0}}$ in the asymptotically flat end. Define
\begin{equation}\label{56.0'}
f_{t,r_{0}} = \left\{
        \begin{array}{lllll}
            \lambda^{2}\eta(\tau) & \quad \tau<\frac{5}{4}\tau_{0} \\
            |E_{t}|_{g_{t}}^{2}+|B_{t}|_{g_{t}}^{2} & \quad \text{ on } M(2\tau_{0},\frac{1}{4}r_{0})\\
             \frac{\delta^{2}}{r_{0}^{4}}\chi(r)& \quad \frac{1}{2}r_{0}<r
        \end{array}
    \right\},
\end{equation}
which agrees with $f_{t}$ on $M(0,\frac{1}{4}r_{0})$.
Here $\delta>0$ is a small parameter to be determined and
$\chi$ is a smooth cut-off function with $\chi\equiv 1$ on
$\frac{1}{2}r_{0}< r< r_{0}+1$, $\chi\equiv 0$ on $r> 2r_{0}$,
$|\nabla\chi|\leq cr_{0}^{-1}$, and $|\nabla^{2}\chi|\leq cr_{0}^{-2}$.
On the transition region $\frac{1}{4}r_{0}<r<\frac{1}{2}r_{0}$, $f_{t,r_{0}}$ is chosen so as to make a
smooth positive function. Next, solve \eqref{54}, \eqref{55} with
$f_{t}$ replaced by $f_{t,r_{0}}$
\begin{equation}\label{54'}
\Delta_{g_{t}} \overline{v}_{t,r_{0}}-f_{t,r_{0}}\overline{v}_{t,r_{0}}=0\text{ }\text{ }\text{ on }\text{ }\text{ }M_{t},
\text{ }\text{ }\text{ }\text{ }\overline{v}_{t,r_{0}}=0\text{ }\text{ }\text{ on }\text{ }\text{ }\partial M_{t},
\end{equation}
\begin{equation}\label{55'}
\overline{v}_{t,r_{0}}=-1+\frac{\overline{\gamma}_{t,r_{0}}}{r}+O\left(\frac{1}{r^{2}}\right)\text{ }\text{ }
\text{ }\text{ as }\text{ }\text{ }\text{ }r\rightarrow\infty.
\end{equation}
The first main task is to establish existence of a positive solution to the auxiliary Dirichlet problem
\begin{equation}\label{1.2'}
\Delta_{g_{t}}\phi_{t,r_{0}}-\frac{\nabla \overline{v}_{t,r_{0}}\cdot\nabla\phi_{t,r_{0}}}{\overline{v}_{t,r_{0}}}
=\frac{\Lambda}{\phi_{t,r_{0}}}\left|f_{t,r_{0}}-\frac{|\nabla\phi_{t,r_{0}}|_{g_{t}}^{2}}
{\overline{v}_{t,r_{0}}^{2}}\right|\text{ }\text{ }\text{ on }\text{ }\text{ }M(\tau_{0},r_{0}),
\end{equation}
\begin{equation}\label{1.3'}
\phi_{t}=\lambda\tau_{0}^{4}\text{ }\text{ }\text{ on }\text{ }\text{ }S_{\tau_{0}},\text{ }\text{ }\text{ }\text{ }\text{ }
\phi_{t}=\frac{\delta}{4r_{0}}\text{ }\text{ }\text{ on }\text{ }\text{ }S_{r_{0}}.
\end{equation}
A priori estimates will be shown to hold independent of $r_{0}$, so that the desired solution of \eqref{1.2}, \eqref{1.3} will arise as the limit $\phi_{t,r_{0}}\rightarrow \phi_{t}$ as $r_{0}\rightarrow\infty$.

The following version of the Leray-Schauder fixed point theorem will be applied to \eqref{1.2'}, \eqref{1.3'}. In what follows, we will temporarily drop all references to the subindices $t$ and $r_{0}$ associated with functions, as well as the subscript $g_{t}$ associated with operators and norms.

\begin{theorem} \label{theorem:lerayschauder}
Suppose that $\mathcal{B}$ is a Banach space with norm $\parallel\cdot\parallel$, $\mathcal{C}\subset\mathcal{B}$ is a closed convex subset, $\phi_{0}$ is a point of
$\mathcal{C}$, $T\colon\mathcal{C}\times[0,1]\to\mathcal{C}$ is continuous and compact with $T(\phi,0)=\phi_{0}$, for all $\phi\in\mathcal{C}$, and suppose that
there is a fixed constant $\Gamma > 0$ such that
\begin{equation} \label{apriori}
  \parallel \phi\parallel < \Gamma
\end{equation}
is satisfied whenever $\phi\in\mathcal{C}$ satisfies $T(\phi,s)=\phi$ for some $s\in[0,1]$. Then there exists $\phi\in\mathcal{C}$ such that $T(\phi,1)=\phi$.
\end{theorem}

Note that this version of the theorem is slightly more general than that given in \cite{GilbargTrudinger}
(Theorem 11.6). The only difference is that $T$ is defined on $\mathcal{C}$ instead of $\mathcal{B}$. This generalization is easily obtained using Dugundji's extension theorem \cite{Deimling} (Theorem 7.2).

In order to set up the fixed point theorem, fix a positive function $\varrho\in C^{\infty}(M)$ with $\varrho\sim r^{-3}$ as $r\rightarrow\infty$, and consider the regularized equation
\begin{equation}\label{3}
\Delta\phi-\frac{\nabla \overline{v}\cdot\nabla\phi}{\overline{v}}=\frac{s\Lambda\phi}{(s\phi+\varepsilon)^{2}}
\left|f-\frac{|\nabla\phi|^{2}}{\overline{v}^{2}}\right|+(1-s)\varrho\phi
\text{ }\text{ }\text{ on }\text{ }\text{ }M(\tau_{0},r_{0}),
\end{equation}
\begin{equation}\label{4}
\phi=\lambda\tau_{0}^{4}\text{ }\text{ }\text{ on }\text{ }\text{ }S_{\tau_{0}},\text{ }\text{ }\text{ }\text{ }\text{ }
\phi=\frac{\delta}{4r_{0}}\text{ }\text{ }\text{ on }\text{ }\text{ }S_{r_{0}}.
\end{equation}
In equation \eqref{3}, there are actually two regularizations at play. One is the $\varepsilon$-regularization which avoids problems when $\phi$ vanishes, and the other is a capillarity regularization associated
with the extra term $(1-s)\varrho\phi$, which aids in establishing $C^{1}$ estimates when $s$ is sufficiently far away from the important value $1$. Later we will let $\varepsilon\rightarrow 0$ in order to obtain a solution of \eqref{1.2'}.

Let $\mathcal{C}$ be the cone of nonnegative $C^{2}(M(\tau_{0},r_{0}))$ functions; note that since $M(\tau_{0},r_{0})$ is closed, this is the space of functions which are $C^{2}$ up to the boundary. It is clear that $\mathcal{C}$ is closed and convex. Define the map $T:\mathcal{C}\times[0,1]\rightarrow\mathcal{C}$ by $T(\phi,s)=\psi$, where $\psi$ solves
\begin{equation}\label{5}
\Delta\psi-\frac{\nabla \overline{v}\cdot\nabla\psi}{\overline{v}}
-\left(\frac{s\Lambda}{(s\phi+\varepsilon)^{2}}\left|f-\frac{|\nabla\phi|^{2}}{\overline{v}^{2}}\right|
+(1-s)\varrho\right)\psi=0
\text{ }\text{ }\text{ on }\text{ }\text{ }M(\tau_{0},r_{0}),
\end{equation}
\begin{equation}\label{6}
\psi=\lambda\tau_{0}^{4}\text{ }\text{ }\text{ on }\text{ }\text{ }S_{\tau_{0}},\text{ }\text{ }\text{ }\text{ }\text{ }
\psi=\frac{\delta}{4r_{0}}\text{ }\text{ }\text{ on }\text{ }\text{ }S_{r_{0}}.
\end{equation}
Then given $\phi\in\mathcal{C}$ there exists a unique solution $\psi\in C^{2,\alpha}(M(\tau_{0},r_{0}))$ by elliptic theory, for any $0<\alpha<1$. Moreover $\psi>0$ by the maximum principle, so $\psi\in\mathcal{C}$.
The Schauder estimates imply that
\begin{equation}\label{7}
|\psi|_{C^{2,\alpha}(M(\tau_{0},r_{0}))}\leq C(|\phi|_{C^{1,\alpha}(M(\tau_{0},r_{0}))},\varepsilon),
\end{equation}
and that $T$ is continuous. Since $C^{2,\alpha}\hookrightarrow C^{2}$ is compact, we also find that $T$ is compact.

Next observe that if $s=0$ then $\psi$ does not depend on $\phi$. Thus, in order to apply the Leray-Schauder fixed point theorem, it remains only to prove the a priori estimate
\begin{equation}\label{8}
|\phi|_{C^{2,\alpha}(M(\tau_{0},r_{0}))}\leq C
\end{equation}
for a fixed point $T(\phi,s)=\phi$, where $C$ is independent of $s$. Note that a fixed point satisfies the boundary value problem \eqref{3}, \eqref{4}. The estimate \eqref{8} will be established in several steps.
First, maximum principle techniques produce $C^{0}$ bounds and also reduce $C^{1}$ bounds to boundary gradient estimates, which are then obtained with a local barrier argument. A positive subsolution is then
constructed, which allows a boot-strap procedure to yield higher order bounds.

\begin{proposition}\label{supbounds}
For any $s$ and $\varepsilon$, $\sup_{M(\tau_{0},r_{0})}\phi\leq\max\{\lambda\tau_{0}^{4},\frac{\delta}{4r_{0}}\}$.
\end{proposition}

\begin{proof}
This follows directly from the maximum principle applied to \eqref{3}, \eqref{4}.
\end{proof}

\begin{proposition}\label{interiorgrad}
If $\Lambda>8$ then
there exists a constant $C$ independent of $s$, $\varepsilon$, and $r_{0}$ such that
\begin{equation}\label{117}
\sup_{M(\tau_{0},r_{0})}|\nabla\phi|\leq C(1+|f|_{C^{1}})+\sup_{\partial M(\tau_{0},r_{0})}|\nabla\phi|.
\end{equation}
\end{proposition}

\begin{proof}
We will apply a maximum principle argument to the equation satisfied by $|\nabla \phi|$. Observe that
\begin{align}\label{12}
\begin{split}
\Delta|\nabla \phi|&=\nabla^{j}\left(\frac{\phi^{i}\nabla_{ij}\phi}{|\nabla \phi|}\right)\\
&=\frac{|\nabla^{2}\phi|^{2}}{|\nabla \phi|}+\frac{\phi^{i}\nabla_{j}\nabla_{i}\nabla_{j}\phi}{|\nabla \phi|}
-\frac{\phi^{i}(\nabla_{ij}\phi)\phi^{l}\nabla_{l}^{j}\phi}{|\nabla \phi|^{3}}\\
&=\frac{|\nabla^{2}\phi|^{2}}{|\nabla \phi|}+\frac{\phi^{i}\partial_{i}\Delta \phi}{|\nabla \phi|}
+\frac{R_{ij}\phi^{i}\phi^{j}}{|\nabla \phi|}-\frac{|\nabla|\nabla \phi||^{2}}{|\nabla \phi|},
\end{split}
\end{align}
where $R_{ij}$ denotes components of the Ricci tensor. Suppose that a global interior maximum exists for $|\nabla \phi|$. Then at this point we may assume that
\begin{equation}\label{13}
\frac{|\nabla\phi|^{2}}{\overline{v}^{2}}>\frac{C^{2}(1+|f|_{C^{1}})^{2}}{\overline{v}^{2}}\geq 2f
\end{equation}
for some constant $C>0$, otherwise the desired result holds immediately.
Thus by \eqref{3} and \eqref{12}, and setting $h(\phi)=s\Lambda\phi(s\phi+\varepsilon)^{-2}$,
it follows that
\begin{align}\label{15}
\begin{split}
\Delta|\nabla \phi|&=\frac{|\nabla^{2}\phi|^{2}}{|\nabla \phi|}
+\frac{1}{|\nabla \phi|}\left[\frac{\phi^{i}\nabla_{il}\overline{v}\nabla^{l}\phi}{\overline{v}}
+\frac{\phi^{i}\nabla_{l}\overline{v}\nabla_{i}^{l}\phi}{\overline{v}}
-\frac{(\nabla \overline{v}\cdot\nabla \phi)^{2}}{\overline{v}^{2}}+h'|\nabla \phi|^{2}
\left(\frac{|\nabla \phi|^{2}}{\overline{v}^{2}}-f\right)\right]\\
&+\frac{h}{|\nabla \phi|}\left(\frac{\phi^{i}\partial_{i}|\nabla \phi|^{2}}{\overline{v}^{2}}
-2\frac{(\nabla \phi\cdot\nabla \overline{v})}{\overline{v}^{3}}|\nabla \phi|^{2}
-\nabla \phi\cdot\nabla f\right)+\frac{R_{ij}\phi^{i}\phi^{j}}{|\nabla \phi|}-\frac{|\nabla|\nabla \phi||^{2}}{|\nabla \phi|}+(1-s)\frac{\phi^{i}\partial_{i}\left(\varrho\phi\right)}{|\nabla\phi|}.
\end{split}
\end{align}
At the maximum
\begin{equation}\label{16}
0=\partial_{i}|\nabla \phi|^{2}=2\phi_{j}\nabla_{i}^{j}\phi,\text{ }\text{ }\text{ }\text{ }\text{ }\Delta|\nabla \phi|\leq 0.
\end{equation}
Hence
\begin{equation}\label{17}
0\geq\frac{|\nabla^{2}\phi|^{2}}{|\nabla \phi|}
+\frac{h'}{\overline{v}^{2}}|\nabla \phi|^{3}-\frac{2h|\nabla\overline{v}|}{|\overline{v}|^{3}}|\nabla\phi|^{2}
+\left((1-s)\varrho-|h'|f-c(x)\right)|\nabla\phi|-h|\nabla f|-(1-s)|\nabla\varrho|\phi,
\end{equation}
for some positive function $c(x)$ independent of $s$, $\varepsilon$, and $r_{0}$, and falling-off at least on the order of $r^{-3}$ in the asymptotic end.

Let $\{e_{1},e_{2},e_{3}=\frac{\nabla\phi}{|\nabla\phi|}\}$ be an orthonormal basis of tangent vectors at the maximum point, then with the help of \eqref{16}
\begin{equation}
|\nabla^{2}\phi|^{2}=\sum_{i,j=1,2}\left[\nabla^{2}\phi(e_{i},e_{j})\right]^{2}
\geq\frac{1}{2}\left[\sum_{i=1,2}\nabla^{2}\phi(e_{i},e_{i})\right]^{2}
=\frac{1}{2}\left(\Delta\phi\right)^{2}.
\end{equation}
Moreover, using \eqref{3} and \eqref{13} produces
\begin{equation}\label{00000}
\left(\Delta\phi\right)^{2}
\geq \frac{h^{2}}{4\overline{v}^{4}}|\nabla\phi|^{4}-\frac{2h|\nabla\overline{v}|}{|\overline{v}|^{3}}|\nabla\phi|^{3}
-\frac{2(1-s)\varrho\phi|\nabla\overline{v}|}{|\overline{v}|}|\nabla\phi|
-\frac{2hf|\nabla\overline{v}|}{|\overline{v}|}|\nabla\phi|.
\end{equation}
Combining this with \eqref{17} then yields
\begin{align}
\begin{split}
0\geq&
\left(\frac{h^{2}}{8\overline{v}^{4}}+\frac{h'}{\overline{v}^{2}}\right)|\nabla \phi|^{3}-\frac{3h|\nabla\overline{v}|}{|\overline{v}|^{3}}|\nabla\phi|^{2}
+\left((1-s)\varrho-|h'|f-c(x)\right)|\nabla\phi|\\
&-h|\nabla f|-(1-s)|\nabla\varrho|\phi
-\frac{(1-s)\varrho\phi|\nabla\overline{v}|}{|\overline{v}|}
-\frac{hf|\nabla\overline{v}|}{|\overline{v}|}.
\end{split}
\end{align}

Let us now calculate
\begin{align}
\begin{split}
h'+\frac{1}{8}h^{2}&=\frac{s\Lambda(-s\phi+\varepsilon)}{(s\phi+\varepsilon)^{3}}
+\frac{s^{2}\Lambda^{2}\phi^{2}}{8(s\phi+\varepsilon)^{4}}\\
&=\frac{\Lambda}{8}\left[\frac{(\Lambda-8s)(s\phi+\varepsilon)^{2}
+\varepsilon(16s-2\Lambda)(s\phi+\varepsilon)+\Lambda\varepsilon^{2}}{(s\phi+\varepsilon)^{4}}\right].
\end{split}
\end{align}
By Young's inequality
\begin{equation}
\varepsilon(16s-2\Lambda)(s\phi+\varepsilon)
\leq\Lambda\varepsilon^{2}+\frac{(8s-\Lambda)^{2}(s\phi+\varepsilon)^{2}}{\Lambda},
\end{equation}
and hence
\begin{equation}
h'+\frac{1}{8}h^{2}\geq\frac{s(\Lambda-8s)}{(s\phi+\varepsilon)^{2}}.
\end{equation}
Then since $|\overline{v}|\leq 1$,
\begin{equation}\label{asdfgh}
\frac{h^{2}}{8\overline{v}^{4}}+\frac{h'}{\overline{v}^{2}}
\geq\frac{s(\Lambda-8s)}{(s\phi+\varepsilon)^{2}\overline{v}^{2}}.
\end{equation}

We are now in a position to obtain a contradiction to the assumption \eqref{13}, if $C$ is chosen sufficiently
large and independent of $s$, $\varepsilon$, and $r_{0}$. If $s\leq\frac{1}{2}$, then apply \eqref{13} and \eqref{asdfgh} to dominate all terms in \eqref{17} involving $h$. That is,
\begin{equation}\label{AS}
\left(\frac{h^{2}}{8\overline{v}^{4}}+\frac{h'}{\overline{v}^{2}}\right)|\nabla \phi|^{3}-\frac{3h|\nabla\overline{v}|}{|\overline{v}|^{3}}|\nabla\phi|^{2}
-|h'|f|\nabla\phi|-h|\nabla f|
-\frac{hf|\nabla\overline{v}|}{|\overline{v}|}>0
\end{equation}
if $C$ is large enough. Similarly, if $s\geq\frac{1}{2}$, then \eqref{13} and \eqref{asdfgh} may be used to dominate all terms in \eqref{17} whether or not they involve $h$. Furthermore, for $s\leq\frac{1}{2}$, $(1-s)\varrho|\nabla\phi|$ may be used to dominate all terms not involving $h$.
That is,
\begin{equation}\label{AD}
\left((1-s)\varrho-c(x)\right)|\nabla\phi|-(1-s)|\nabla\varrho|\phi
-\frac{(1-s)\varrho\phi|\nabla\overline{v}|}{|\overline{v}|}>0
\end{equation}
if $C$ is large enough. Notice that inequalities \eqref{AS} and \eqref{AD} are not consistent with \eqref{17}.
We conclude that there must exist a finite $C$, independent of $s$, $\varepsilon$, and $r_{0}$, such that
$|\nabla\phi|\leq C(1+|f|_{C^{1}})$ at a global interior max, if this point exists. If the global maximum is not attained on the interior, then it must be obtained on the boundary, and the desired result \eqref{117} follows.
\end{proof}

We will now establish boundary gradient estimates by constructing appropriate local barriers.

\begin{lemma}\label{lemmagrad}
If $\lambda\tau_{0}^{4}\geq\frac{\delta}{4r_{0}}$, and $\tau_{0}>0$ is sufficiently small, then there exists a constant $C$ independent of $s$, $\varepsilon$, $\lambda$, and $\tau_{0}$ such that
\begin{equation}
|\nabla\phi|_{S_{\tau_{0}}}\leq C\lambda\tau_{0}.
\end{equation}
\end{lemma}

\begin{proof}
Since $\phi$ is constant on $S_{\tau_{0}}$ it suffices to estimate the normal derivative.
As $\lambda\tau_{0}^{4}\geq\frac{\delta}{4r_{0}}$, an upper barrier is trivial to construct. Namely, by the maximum principle $\phi\leq\lambda\tau_{0}^{4}$ globally. Hence we have
\begin{equation}\label{53}
\partial_{\tau}\phi|_{S_{\tau_{0}}}\leq 0.
\end{equation}

A lower barrier will now be constructed as the solution to an Eikonal equation near the boundary
\begin{equation}\label{59}
|\nabla\underline{\phi}|^{2}=f\overline{v}^{2}\text{ }\text{ }\text{ }\text{ on }\text{ }\text{ }\text{ }D(\tau_{0},\tau_{1}),\text{ }\text{ }\text{ }\text{ }\text{ }
\underline{\phi}=\lambda\tau_{0}^{4}\text{ }\text{ }\text{ }\text{ on }\text{ }\text{ }\text{ }S_{\tau_{0}},
\end{equation}
where $\tau_{1}=\tau_{0}+\tau_{0}^{5/2}$ and $D(\tau_{0},\tau_{1})$ denotes the domain enclosed by $S_{\tau_{0}}$ and $S_{\tau_{1}}$. Note that the surface $S_{\tau_{0}}$ is noncharacteristic for this initial value problem, since it is possible to solve for $\partial_{\tau}\underline{\phi}|_{S_{\tau_{0}}}$.
We claim that there is a solution with $\underline{\phi}<0$ on $S_{\tau_{1}}$. This will follow from an implicit function theorem argument. First construct an approximate solution $\phi_{0}$. Expand
\begin{equation}\label{61}
\overline{v}=\overline{v}_{0}+\overline{v}_{1}(\tau-\tau_{0})+O(|\tau-\tau_{0}|^{2}),
\end{equation}
and observe that since $\partial_{\tau}\overline{v}|_{\partial M}<0$ we have $\partial_{\tau}\overline{v}|_{S_{\tau_{0}}}=\overline{v}_{1}<0$, and also $-c\tau_{0}\leq\overline{v}_{0}\leq-c^{-1}\tau_{0}$. Then plugging \eqref{61} into equation \eqref{59}, and using $f=\lambda^{2}[1-2\tau_{0}^{-1}(\tau-\tau_{0})]$ on $D(\tau_{0},\tau_{1})$, yields
\begin{equation}\label{66}
\phi_{0}=\lambda\tau_{0}^{4}+\lambda\overline{v}_{0}(\tau-\tau_{0})
+\frac{\lambda}{2}\left(\overline{v}_{1}-\tau_{0}^{-1}\overline{v}_{0}\right)(\tau-\tau_{0})^{2}
+\cdots+O\left(\frac{\lambda|\tau-\tau_{0}|^{N+1}}{\tau_{0}^{N-1}}\right),
\end{equation}
with
\begin{equation}\label{67}
|\nabla\phi_{0}|^{2}=f\overline{v}^{2}+O\left(\frac{\lambda^{2}|\tau-\tau_{0}|^{N}}{\tau_{0}^{N-2}}\right)
\end{equation}
for any large $N$, depending on how many terms are given in $\phi_{0}$. Below, $N\geq 3$ is sufficient.

Consider now the linearized equation
\begin{equation}\label{68}
\mathcal{L}\underline{\varphi}:=\nabla\phi_{0}\cdot\nabla\underline{\varphi}=\psi\text{ }\text{ }\text{ }\text{ on }\text{ }\text{ }\text{ }D(\tau_{0},\tau_{1}),\text{ }\text{ }\text{ }\text{ }\text{ }
\underline{\varphi}=0\text{ }\text{ }\text{ }\text{ on }\text{ }\text{ }\text{ }S_{\tau_{0}}.
\end{equation}
Since $\partial_{\tau}\phi_{0}|_{S_{\tau_{0}}}=\lambda\overline{v}_{0}\neq 0$, the method of characteristics shows that we can always solve this initial value problem. Moreover $\mathcal{L}:\widetilde{C}^{1}\rightarrow C^{0}$ is an isomorphism if $\widetilde{C}^{1}$ consists of all $C^{1}$ functions vanishing on $S_{\tau_{0}}$. We may now apply the implicit function theorem to obtain a local smooth solution $\underline{\phi}$ of \eqref{59} for $|\tau_{0}-\tau_{1}|=\tau_{0}^{5/2}$ sufficiently small.

To show that $\underline{\phi}|_{S_{\tau_{1}}}<0$, choose $\tau_{0}$ sufficiently small so that $\phi_{0}|_{S_{\tau_{1}}}<0$. This is satisfied if
\begin{equation}\label{68.0}
\lambda\tau_{0}^{4}<\lambda|\overline{v}_{0}||\tau_{1}-\tau_{0}|\sim\lambda\tau_{0}^{7/2}.
\end{equation}
It follows that
\begin{equation}\label{69.1} \underline{\phi}=\phi_{0}+(\underline{\phi}-\phi_{0})=\phi_{0}
+O\left(\frac{\lambda|\tau-\tau_{0}|^{N}}{\tau_{0}^{N-2}}\right)
<0\text{ }\text{ }\text{ on }\text{ }\text{ }S_{\tau_{1}},
\end{equation}
if $\tau_{0}$ is sufficiently small.

Continuing with the proof of the boundary gradient estimate, we will use $\underline{\phi}$ as a lower barrier. To see that $\underline{\phi}$ is a subsolution of \eqref{3} on $D(\tau_{0},\tau_{1})$, calculate
\begin{align}\label{70}
\begin{split}
\Delta\underline{\phi}&=\partial_{\tau}^{2}\underline{\phi}+H\partial_{\tau}\underline{\phi}+\Delta_{S_{\tau}}\underline{\phi}\\
&=\partial_{\tau}^{2}\phi_{0}+H\partial_{\tau}\phi_{0}+\Delta_{S_{\tau}}\phi_{0}
+O\left(\frac{\lambda|\tau_{1}-\tau_{0}|^{N-2}}{\tau_{0}^{N-2}}\right)\\
&=\lambda(\overline{v}_{1}-\tau_{0}^{-1}\overline{v}_{0})
+O\left(\frac{\lambda|\tau_{1}-\tau_{0}|}{\tau_{0}}\right),
\end{split}
\end{align}
and
\begin{align}\label{71}
\begin{split}
\frac{\nabla \overline{v}\cdot\nabla\underline{\phi}}{\overline{v}}&=
\frac{\partial_{\tau}\overline{v}\partial_{\tau}\underline{\phi}}{\overline{v}}
+\frac{\overline{\nabla}\overline{v}\cdot\overline{\nabla}\underline{\phi}}{\overline{v}}\\
&=\frac{\partial_{\tau}\overline{v}\partial_{\tau}\phi_{0}}{\overline{v}}
+\frac{\overline{\nabla}\overline{v}\cdot\overline{\nabla}\phi_{0}}{\overline{v}}
+O\left(\frac{\lambda|\tau_{1}-\tau_{0}|^{N-1}}{\tau_{0}^{N-1}}\right)\\
&=\left(\frac{1}{\tau}+O(1)\right)\left(\lambda\overline{v}_{0}+O(\lambda|\tau_{1}-\tau_{0}|)\right)
+O\left(\frac{\lambda|\tau_{1}-\tau_{0}|}{\tau_{0}}\right)\\
&=\frac{\lambda\overline{v}_{0}}{\tau}
+O\left(\lambda|\overline{v}_{0}|+\frac{\lambda|\tau_{1}-\tau_{0}|}{\tau_{0}}\right),
\end{split}
\end{align}
where we have used $\overline{v}=\widetilde{v}\tau+O(\tau^{2})$ with $\widetilde{v}<0$; $\overline{\nabla}$ represents the induced connection on $S_{\tau}$. On the other side of \eqref{3}, equation \eqref{59} eliminates one term and the other satisfies
\begin{equation}\label{73}
(1-s)\varrho\underline{\phi}\leq c\lambda\tau_{0}^{3}.
\end{equation}

Now observe that
\begin{equation}\label{73.0}
\overline{v}_{1}=\partial_{\tau}\overline{v}|_{S_{\tau_{0}}}=\partial_{\tau}\overline{v}|_{\partial M}+O(\tau_{0}),
\text{ }\text{ }\text{ }\text{ }\text{ }
\overline{v}_{0}=\overline{v}|_{S_{\tau_{0}}}=\partial_{\tau}\overline{v}|_{\partial M}\tau_{0}+O(\tau_{0}^{2}),
\end{equation}
which implies
\begin{equation}\label{73.1}
\overline{v}_{1}-\left(\frac{1}{\tau_{0}}+\frac{1}{\tau}\right)\overline{v}_{0}
=-\frac{\tau_{0}}{\tau}\partial_{\tau}\overline{v}|_{\partial M}+O(\tau_{0})
\geq 2c+O\left(\frac{|\tau_{1}-\tau_{0}|}{\tau_{0}}\right)
\end{equation}
for some constant $c>0$. Hence, if $\tau_{0}$ is sufficiently small
\begin{equation}\label{79}
\Delta\underline{\phi}-\frac{\nabla \overline{v}\cdot\nabla\underline{\phi}}{\overline{v}}=
\lambda\left[\overline{v}_{1}-\left(\frac{1}{\tau_{0}}+\frac{1}{\tau}\right)\overline{v}_{0}\right]
+O\left(\lambda|\overline{v}_{0}|+\frac{\lambda|\tau_{1}-\tau_{0}|}{\tau_{0}}\right)\\
\geq 2c\lambda+O(\lambda\tau_{0}^{3/2})>0.
\end{equation}
Note that this positive lower bound is a result of the choice
$f=\lambda^{2}[1-2\tau_{0}^{-1}(\tau-\tau_{0})]$ on $D(\tau_{0},\tau_{1})$, and is the reason for defining $f$ in this way. In light of \eqref{59} and \eqref{73}, we find that $\underline{\phi}$ is a subsolution of \eqref{3} in the following sense
\begin{equation}\label{79.1}
\Delta\underline{\phi}-\frac{\nabla \overline{v}\cdot\nabla\underline{\phi}}{\overline{v}}
>\frac{s\Lambda\phi}{(s\phi+\varepsilon)^{2}}\left|f-\frac{|\nabla\underline{\phi}|^{2}}{\overline{v}^{2}}\right|
+(1-s)\varrho\underline{\phi}.
\end{equation}
It is important that the coefficient $\frac{s\Lambda\phi}{(s\phi+\varepsilon)^{2}}$ involves $\phi$ instead of $\underline{\phi}$, for the comparison argument below. This motivates using the Eikonal equation \eqref{59} to construct the subsolution.

Having shown that $\underline{\phi}$ is a subsolution, we now compare it with $\phi$. Note that
\begin{equation}\label{80}
\Delta(\phi-\underline{\phi})-\frac{\nabla \overline{v}\cdot\nabla(\phi-\underline{\phi})}{\overline{v}}
<\frac{s\Lambda\phi}{(s\phi+\varepsilon)^{2}}\left(\left|f-\frac{|\nabla\phi|^{2}}{\overline{v}^{2}}\right|
-\left|f-\frac{|\nabla\underline{\phi}|^{2}}{\overline{v}^{2}}\right|\right)
+(1-s)\varrho(\phi-\underline{\phi}),
\end{equation}
and
\begin{equation}\label{81}
(\phi-\underline{\phi})|_{S_{\tau_{0}}}=0,\text{ }\text{ }\text{ }\text{ }\text{ }\text{ }
(\phi-\underline{\phi})|_{S_{\tau_{1}}}>0.
\end{equation}
At an interior minimum of $\phi-\underline{\phi}$, $|\nabla\phi|=|\nabla\underline{\phi}|$
and $\Delta(\phi-\underline{\phi})\geq 0$, which yields a contradiction.
We conclude that $\phi-\underline{\phi}\geq 0$ on $D(\tau_{0},\tau_{1})$, and hence $\partial_{\tau}\phi|_{S_{\tau_{0}}}\geq\partial_{\tau}\underline{\phi}|_{S_{\tau_{0}}}$.
Since
\begin{equation}\label{83}
\partial_{\tau}\underline{\phi}=\partial_{\tau}\phi_{0}
+O\left(\frac{\lambda|\tau_{1}-\tau_{0}|^{N-1}}{\tau_{0}^{N-1}}\right)
=\lambda\overline{v}_{0}+O\left(\frac{\lambda|\tau_{1}-\tau_{0}|}{\tau_{0}}\right)
\geq-c\lambda\tau_{0},
\end{equation}
the desired result follows.
\end{proof}

Using similar methods a boundary gradient estimate may be established at $S_{r_{0}}$.

\begin{lemma}\label{lemmagrad1}
If $r_{0}$ is sufficiently large, then
\begin{equation}
|\nabla\phi|_{S_{r_{0}}}\leq \frac{2}{r_{0}}.
\end{equation}
\end{lemma}

\begin{proof}
Let $r_{1}=\frac{1}{2}r_{0}$. An upper barrier may be constructed in the form
\begin{equation}\label{85.1}
\overline{\phi}=\frac{\delta}{4r_{0}}+a(r_{0}-r)\text{ }\text{ }\text{ }\text{ for }\text{ }\text{ }\text{ }r\in[r_{1},r_{0}].
\end{equation}
A basic calculation shows that
\begin{equation}\label{88}
\Delta\overline{\phi}-\frac{\nabla \overline{v}\cdot\nabla\overline{\phi}}{\overline{v}}=-\frac{2a}{r}+O\left(\frac{a}{r^{2}}\right).
\end{equation}
Hence, if $a>0$ and $r_{0}$ is large enough, $\overline{\phi}$ is a supersolution
\begin{align}\label{89}
\begin{split}
\Delta\overline{\phi}-\frac{\nabla \overline{v}\cdot\nabla\overline{\phi}}{\overline{v}}&=-\frac{2a}{r}+O\left(\frac{a}{r^{2}}\right)\\
&<0\leq\frac{s\Lambda\phi}{(s\phi+\varepsilon)^{2}}\left|f-\frac{|\nabla\overline{\phi}|^{2}}{\overline{v}^{2}}\right|
+(1-s)\varrho\overline{\phi}.
\end{split}
\end{align}
Moreover if $a=2r_{0}^{-1}$, then
\begin{equation}\label{91}
\overline{\phi}|_{S_{r_{1}}}=\frac{\delta}{4r_{0}}+1>\lambda\tau_{0}^{4}=\sup_{M(\tau_{0},r_{0})}\phi.
\end{equation}
A comparison argument then implies that $\overline{\phi}\geq\phi$ for $r_{1}\leq r\leq r_{0}$. Therefore
\begin{equation}\label{92}
\partial_{r}\phi|_{S_{r_{0}}}\geq-\frac{2}{r_{0}}.
\end{equation}

As in the proof of Lemma \ref{lemmagrad}, a lower barrier will be constructed as a solution to the Eikonal equation
\begin{equation}\label{93}
|\nabla\underline{\phi}|^{2}=f\overline{v}^{2}\text{ }\text{ }\text{ }\text{ on }\text{ }\text{ }\text{ }B(r_{1},r_{0}),
\text{ }\text{ }\text{ }\text{ }\text{ }
\underline{\phi}=\frac{\delta}{4r_{0}}\text{ }\text{ }\text{ }\text{ on }\text{ }\text{ }\text{ }S_{r_{0}},
\end{equation}
where $B(r_{1},r_{0})$ is the region between $S_{r_{1}}$ and $S_{r_{0}}$. Note that $S_{r_{0}}$ is noncharacteristic for this initial value problem, since it is possible to solve for $\partial_{r}\underline{\phi}|_{S_{r_{0}}}$. In order to apply the implicit function theorem, an approximate solution may be found in the form
\begin{equation}\label{95}
\phi_{0}=\frac{\delta}{4r_{0}}-\frac{\delta}{r_{0}^{2}}(r_{0}-r).
\end{equation}
Using the fact that $f=\frac{\delta^{2}}{r_{0}^{4}}$, a direct calculation produces
\begin{equation}\label{98}
|\nabla\phi_{0}|^{2}-f\overline{v}^{2}=O\left(\frac{\delta^{2}}{rr_{0}^{4}}\right).
\end{equation}
The implicit function theorem now yields a solution
\begin{equation}\label{108}
\underline{\phi}=\phi_{0}
+O\left(\frac{\delta}{r_{0}^{2}}\right).
\end{equation}
Moreover
\begin{equation}\label{109}
\underline{\phi}|_{S_{r_{1}}}=-\frac{\delta}{4r_{0}}+O\left(\frac{\delta}{r_{0}^{2}}\right)<0,
\end{equation}
if $r_{0}$ is large enough.

In order to show that $\underline{\phi}$ is a subsolution on $B(r_{1},r_{0})$, observe that
\begin{equation}\label{110}
\Delta\underline{\phi}=\partial_{r}^{2}\underline{\phi}+\frac{2}{r}\partial_{r}\underline{\phi}
+O(r^{-1}|\nabla^{2}\underline{\phi}|+r^{-2}|\nabla\underline{\phi}|)
=\frac{2\delta}{rr_{0}^{2}}+O\left(\frac{\delta}{r_{0}^{4}}\right),
\end{equation}
and
\begin{equation}\label{111}
\frac{\nabla \overline{v}\cdot\nabla\underline{\phi}}{\overline{v}}=O\left(\frac{|\nabla\underline{\phi}|}{r^{2}}\right)
=O\left(\frac{\delta}{r_{0}^{4}}\right).
\end{equation}
On the other side of equation \eqref{3}, one term is eliminated with the aid of \eqref{93} and the other satisfies
\begin{equation}\label{114}
(1-s)\varrho|\underline{\phi}|=O\left(\frac{\delta}{r_{0}^{4}}\right),
\end{equation}
since $\varrho\leq cr^{-3}$. It follows that
\begin{align}\label{115}
\begin{split}
\Delta\underline{\phi}-\frac{\nabla \overline{v}\cdot\nabla\underline{\phi}}{\overline{v}}
&=\frac{2\delta}{rr_{0}^{2}}+O\left(\frac{\delta}{r_{0}^{4}}\right)\\
&\geq \frac{2\delta}{r_{0}^{3}}+O\left(\frac{\delta}{r_{0}^{4}}\right)\\
&>O\left(\frac{\delta}{r_{0}^{4}}\right)
=\frac{s\Lambda\phi}{(s\phi+\varepsilon)^{2}}\left|f-\frac{|\nabla\underline{\phi}|^{2}}{\overline{v}^{2}}\right|
+(1-s)\varrho\underline{\phi}.
\end{split}
\end{align}
A comparison argument now shows that $\phi\geq\underline{\phi}$ on $B(r_{1},r_{0})$, which yields
\begin{equation}\label{115.0}
\partial_{r}\phi|_{S_{r_{0}}}\leq\frac{\delta}{2r_{0}^{2}}
\end{equation}
if $r_{0}$ is sufficiently large.
\end{proof}

By combining Proposition \ref{interiorgrad}, Lemma \ref{lemmagrad}, and Lemma \ref{lemmagrad1} we obtain global $C^1$ bounds.

\begin{cor}\label{globalgrad}
If $\tau_{0}>0$ is sufficiently small, $r_{0}$ is sufficiently large, and $\Lambda>8$ then
there exists a constant $C$ independent of $s$, $\varepsilon$, and $r_{0}$ such that
\begin{equation}\label{115.1}
\sup_{M(\tau_{0},r_{0})}|\nabla\phi|\leq C(1+|f|_{C^{1}}).
\end{equation}
\end{cor}

We are now in position to establish a basic existence results.

\begin{theorem}\label{existence0}
If $\lambda,\tau_{0}>0$, $\tau_{0}$ is sufficiently small, $r_{0}$ is sufficiently large, and $\Lambda>8$ then
there exists a positive solution $\phi_{r_{0},\varepsilon}\in C^{2,\alpha}(M(\tau_{0},r_{0}))$, for any $\alpha\in[0,1)$, of \eqref{3}, \eqref{4} with $s=1$.
\end{theorem}

\begin{proof}
In order to apply the Leray-Schauder theorem, it suffices to establish
a $C^{2,\alpha}$ estimate. Observe that with the aid of Proposition \ref{supbounds} and Corollary \ref{globalgrad}, and the fact that $s\phi+\varepsilon\geq\varepsilon>0$, we have
$\Delta\phi-\frac{\nabla \overline{v}\cdot\nabla\phi}{\overline{v}}\in L^{\infty}$ and the corresponding bound is independent of $s$. This implies that $\phi\in W^{2,p}$ for any $p$, and hence $\phi\in C^{1,\alpha}$ for any $\alpha <1$, again independent of $s$. Thus $\Delta\phi-\frac{\nabla \overline{v}\cdot\nabla\phi}{\overline{v}}\in C^{0,\alpha}$, which implies that $\phi\in C^{2,\alpha}$. That is,
\begin{equation}\label{138.1}
|\phi|_{C^{2,\alpha}(M(\tau_{0},r_{0}))}\leq C
\end{equation}
where $C$ is independent of $s$. This guarantees the existence of a positive solution $\phi_{r_{0},\varepsilon}\in C^{2,\alpha}(M(\tau_{0},r_{0}))$ of \eqref{3}, \eqref{4} with $s=1$, which satisfies \eqref{138.1}.
\end{proof}

In order to proceed further towards the goal of establishing $C^{2,\alpha}$ estimates independent of $\varepsilon$, a uniform positive lower bound for $\phi_{r_{0},\varepsilon}$ is needed. To this end let us rewrite \eqref{3}, with $s=1$, as an equation for $\zeta:=(\phi_{r_{0},\varepsilon}+\varepsilon)^{3}$. Observe that
\begin{equation}\label{120}
\Delta\zeta=3(\phi_{r_{0},\varepsilon}+\varepsilon)^{2}\Delta\phi_{r_{0},\varepsilon}
+6(\phi+\varepsilon)|\nabla\phi_{r_{0},\varepsilon}|^{2}
\end{equation}
and
\begin{equation}\label{121}
\frac{\nabla \overline{v}\cdot\nabla\zeta}{\overline{v}}=\frac{3(\nabla \overline{v}\cdot\nabla\phi_{r_{0},\varepsilon})}{\overline{v}}(\phi_{r_{0},\varepsilon}+\varepsilon)^{2}.
\end{equation}
It follows that
\begin{equation}\label{122}
\Delta\zeta-\frac{\nabla \overline{v}\cdot\nabla\zeta}{\overline{v}}=3\Lambda[(\phi_{r_{0},\varepsilon}+\varepsilon)-\varepsilon]
\left|f-\frac{|\nabla\phi_{r_{0},\varepsilon}|^{2}}{\overline{v}^{2}}\right|
+6|\nabla\phi_{r_{0},\varepsilon}|^{2}\zeta^{\frac{1}{3}}.
\end{equation}
Hence
\begin{equation}\label{123}
\Delta\zeta-\frac{\nabla \overline{v}\cdot\nabla\zeta}{\overline{v}}-\Phi\zeta^{\frac{1}{3}}=-3\varepsilon \Lambda
\left|f-\frac{|\nabla\phi_{r_{0},\varepsilon}|^{2}}{\overline{v}^{2}}\right|\leq 0,
\end{equation}
where
\begin{equation}\label{124}
\Phi=3\Lambda\left|f-\frac{|\nabla\phi_{r_{0},\varepsilon}|^{2}}{\overline{v}^{2}}\right|
+6|\nabla\phi_{r_{0},\varepsilon}|^{2},
\end{equation}
and
\begin{equation}\label{125}
\zeta|_{S_{\tau_{0}}}=(\lambda\tau_{0}^{4}+\varepsilon)^{3},\text{ }\text{ }\text{ }\text{ }\text{ }
\zeta|_{S_{r_{0}}}=\left(\frac{\delta}{4r_{0}}+\varepsilon\right)^{3}.
\end{equation}
According to the $C^{0}$ and $C^{1}$ estimates $\Phi\leq\Phi_{0}$, where $\Phi_{0}$ is a function independent of $\varepsilon$, $r_{0}$, and $\phi_{r_{0},\varepsilon}$.

Suppose that there exists a solution of the Dirichlet problem
\begin{equation}\label{126}
\Delta\underline{\zeta}-\frac{\nabla \overline{v}\cdot\nabla\underline{\zeta}}{\overline{v}}-\Phi_{0}\underline{\zeta}^{\frac{1}{3}}=0\text{ }\text{ }\text{ }\text{ on }\text{ }\text{ }\text{ }M(\tau_{0},r_{0}),
\end{equation}
\begin{equation}\label{127}
\underline{\zeta}=(\lambda\tau_{0}^{4})^{3}\text{ }\text{ }\text{ on }\text{ }\text{ }S_{\tau_{0}},\text{ }\text{ }\text{ }\text{ }\text{ }
\underline{\zeta}=\left(\frac{\delta}{4r_{0}}\right)^{3}\text{ }\text{ }\text{ on }\text{ }\text{ }S_{r_{0}}.
\end{equation}
By the maximum principle the solution is positive. Moreover, it is independent of $\varepsilon$, and $\phi_{r_{0},\varepsilon}$, and $\zeta\geq\underline{\zeta}$. To see this, use a comparison argument. Observe that
\begin{equation}\label{128}
\Delta(\zeta-\underline{\zeta})-\frac{\nabla \overline{v}\cdot\nabla(\zeta-\underline{\zeta})}{\overline{v}}\leq\Phi\zeta^{\frac{1}{3}}-\Phi_{0}\underline{\zeta}^{\frac{1}{3}},
\end{equation}
\begin{equation}\label{129}
(\zeta-\underline{\zeta})|_{S_{\tau_{0}}}>0,\text{ }\text{ }\text{ }\text{ }\text{ }\text{ }
(\zeta-\underline{\zeta})|_{S_{r_{0}}}>0.
\end{equation}
Suppose that $\zeta-\underline{\zeta}<0$ somewhere, and let $x_{0}$ be the point at which the global minimum is achieved. Then $\zeta^{\frac{1}{3}}(x_{0})<\underline{\zeta}^{\frac{1}{3}}(x_{0})$,
$|\nabla(\zeta-\underline{\zeta})(x_{0})|=0$, and $\Delta(\zeta-\underline{\zeta})(x_{0})\geq 0$, so that
\begin{align}\label{130}
\begin{split}
0&\leq\Delta(\zeta-\underline{\zeta})(x_{0})-\frac{\nabla \overline{v}\cdot\nabla(\zeta-\underline{\zeta})}{\overline{v}}(x_{0})\\
&\leq\Phi\zeta^{\frac{1}{3}}(x_{0})-\Phi_{0}\underline{\zeta}^{\frac{1}{3}}(x_{0})\\
&\leq\underline{\zeta}^{\frac{1}{3}}(\Phi-\Phi_{0})(x_{0})<0.
\end{split}
\end{align}
This contradiction implies that $\zeta\geq\underline{\zeta}$ on all of $M(\tau_{0},r_{0})$, and
yields the desired uniform positive lower bound. Hence
$(\phi_{r_{0},\varepsilon}+\varepsilon)^{3}\geq\underline{\zeta}$ or rather $\phi_{r_{0},\varepsilon}\geq\underline{\zeta}^{\frac{1}{3}}-\varepsilon$. Since $\underline{\zeta}>0$ is independent of $\varepsilon$, we find that $\phi_{r_{0},\varepsilon}\geq\frac{1}{2}\underline{\zeta}^{\frac{1}{3}}$ for all sufficiently small $\varepsilon$.

\begin{proposition}\label{lowerbound}
If $\lambda,\tau_{0}>0$ then there exists a smooth positive function $\underline{\zeta}$ independent of  $\varepsilon$, and $\phi_{r_{0},\varepsilon}$, such that $\phi_{r_{0},\varepsilon}\geq\underline{\zeta}^{\frac{1}{3}}-\varepsilon$ on $M(\tau_{0},r_{0})$. In particular,
$\phi_{r_{0},\varepsilon}\geq\frac{1}{2}\underline{\zeta}^{\frac{1}{3}}$ for all sufficiently small $\varepsilon$.
\end{proposition}

\begin{proof}
It only remains to establish the existence of $\underline{\zeta}$. Write the equation as
\begin{equation}\label{131}
\operatorname{div}\left(\frac{\nabla\underline{\zeta}}{\overline{v}}\right)-\frac{\Phi_{0}}{\overline{v}}\underline{\zeta}^{\frac{1}{3}}=0,
\end{equation}
and consider the functional
\begin{equation}\label{132}
I[\zeta]=\int_{M(\tau_{0},r_{0})}\frac{1}{2}\frac{|\nabla\zeta|^{2}}{|\overline{v}|}
+\frac{3}{4}\frac{\Phi_{0}}{|\overline{v}|}\zeta^{\frac{4}{3}}.
\end{equation}
We will minimize it over all $H^{1}(M(\tau_{0},r_{0}))$-functions with the fixed boundary values as in \eqref{127}. To show that this functional is bounded from below, apply H\"{o}lder's inequality and use $\operatorname{Vol}(M(\tau_{0},r_{0}))<\infty$ to find
\begin{equation}\label{133}
\int_{M(\tau_{0},r_{0})}\frac{\Phi_{0}}{|\overline{v}|}\zeta^{\frac{4}{3}}
\leq \sigma\parallel\zeta\parallel_{L^{2}}^{2}+c_{1},
\end{equation}
where $\sigma$ is a small parameter. Moreover, by Poincar\'{e}'s inequality
\begin{align}\label{135}
\begin{split}
\int_{M(\tau_{0},r_{0})}|\nabla\zeta|^{2}&=\int_{M(\tau_{0},r_{0})}|\nabla(\zeta-(\lambda\tau_{0}^{4})^{3})|^{2}\\
&\geq\int_{M(\tau_{0},r_{0})}c_{2}^{-1}|\zeta-(\lambda\tau_{0}^{4})^{3}|^{2}\\
&=\int_{M(\tau_{0},r_{0})}c_{2}^{-1}(\zeta^{2}-2(\lambda\tau_{0}^{4})^{3}\zeta+(\lambda\tau_{0}^{4})^{6})\\
&\geq\int_{M(\tau_{0},r_{0})}2c_{2}^{-1}\zeta^{2}-c_{3}.
\end{split}
\end{align}
It follows that if $\sigma$ is sufficiently small then
\begin{equation}\label{136}
I[\zeta]\geq c_{4}^{-1}\parallel\zeta\parallel_{H^{1}}^{2}-c_{5},
\end{equation}
and hence $I[\zeta]$ is bounded from below.

By direct methods in the calculus of variations, there is a minimizing sequence that has a weakly convergent (in $H^{1}$) subsequence $\zeta_{i}\rightharpoonup\underline{\zeta}$. Then since $H^{1}\hookrightarrow L^{2}$ is compact, $\zeta_{i}\rightarrow\underline{\zeta}$ strongly in $L^{2}$. In particular $\zeta_{i}\rightarrow\underline{\zeta}$ in $L^{\frac{4}{3}}$,
so that
\begin{equation}\label{138}
\int_{M(\tau_{0},r_{0})}\zeta_{i}^{\frac{4}{3}}\rightarrow
\int_{M(\tau_{0},r_{0})}\underline{\zeta}^{\frac{4}{3}}.
\end{equation}
Since the $H^{1}$ norm is weakly lower semicontinuous, it follows that $\underline{\zeta}$ realizes the infimum, and is hence a weak solution.

Now use elliptic regularity. Since $\underline{\zeta}^{\frac{1}{3}}\in L^{6}$, by \eqref{126} we have $\Delta\underline{\zeta}-\frac{\nabla \overline{v}\cdot\nabla\underline{\zeta}}{\overline{v}}\in L^{6}$, which implies that $\underline{\zeta}\in W^{2,6}\hookrightarrow C^{1,\frac{1}{2}}$. This in turn implies $\underline{\zeta}^{\frac{1}{3}}\in C^{0}$, so that $\underline{\zeta}\in W^{2,p}$ for any $p>1$,
and hence $\underline{\zeta}\in C^{1,\alpha}$ for any $\alpha<1$. Since $\underline{\zeta}>0$ it follows that  $\underline{\zeta}^{\frac{1}{3}}\in C^{0,\alpha}$, and hence $\underline{\zeta}\in C^{2,\alpha}$. Boot-strapping then produces $\underline{\zeta}\in C^{\infty}$.
\end{proof}

Having established a uniform lower bound, we may let $\varepsilon\rightarrow 0$ and $r_{0}\rightarrow\infty$ to obtain the main existence result.

\begin{theorem}\label{existence1}
If $\lambda,\tau_{0}>0$, $\tau_{0}$ is sufficiently small, and $\Lambda>8$, then
there exists a positive solution $\phi_{t}\in C^{2,\alpha}(M(\tau_{0}))$, for any $\alpha\in[0,1)$, of \eqref{1.2}, \eqref{1.3}.
\end{theorem}

\begin{proof}
In light of the uniform lower bound in Proposition \ref{lowerbound}, the same arguments appearing in the proof of Theorem \ref{existence0} yield a $C^{2,\alpha}$ estimate \eqref{138.1}, with $C$ independent of $\varepsilon$. Thus, after possibly passing to a subsequence $\phi_{r_{0},\varepsilon}\rightarrow\phi_{t,r_{0}}$ as $\varepsilon\rightarrow 0$, yielding a positive solution of \eqref{1.2'}, \eqref{1.3'} which also satisfies \eqref{138.1}.

We will now let $r_{0}\rightarrow\infty$. In order to control the decay at spatial infinity, we construct a global upper barrier by solving
\begin{equation}\label{138.2}
\Delta\overline{\phi}_{t,r_{0}}-\frac{\nabla\overline{v}_{t,r_{0}}\cdot\overline{\phi}_{t,r_{0}}}{\overline{v}_{t}}=0\text{ }\text{ }\text{ on }\text{ }\text{ }M(\tau_{0}), \text{ }\text{ }\text{ }\text{ }\overline{\phi}_{t,r_{0}}|_{S_{\tau_{0}}}
=1, \text{ }\text{ }\text{ }\text{ }\overline{\phi}_{t,r_{0}}=\frac{c_{t,r_{0}}}{r}+O\left(\frac{1}{r^{2}}\right)\text{ }\text{ }\text{ as }\text{ }\text{ }
r\rightarrow\infty,
\end{equation}
where $c_{t,r_{0}}$ is a positive constant. Since $\overline{v}_{t,r_{0}}$ smoothly converges to $\overline{v}_{t}$ as $r_{0}\rightarrow\infty$, it holds that $c_{t,r_{0}}\rightarrow c_{t}>0$ and $\overline{\phi}_{t,r_{0}}\rightarrow\overline{\phi}_{t}$, the solution of \eqref{138.2} with $\overline{v}_{t,r_{0}}$ replaced by $\overline{v}_{t}$. Thus
if $\delta$ is chosen sufficiently small (independent of $r_{0}$), then $\overline{\phi}_{t,r_{0}}|_{S_{r_{0}}}>\phi_{t,r_{0}}|_{S_{r_{0}}}$. A standard comparison argument now applies to yield $\overline{\phi}_{t,r_{0}}>\phi_{t,r_{0}}$ on $M(\tau_{0},r_{0})$.

Let $\underline{\zeta}_{t,r_{0}}$ denote the lower bound which arises from the construction in Proposition \ref{lowerbound}. Since $\Phi_{0}$ and $\sup\underline{\zeta}_{t,r_{0}}$ are controlled independent of $r_{0}$, higher order a priori estimates for $\underline{\zeta}_{t,r_{0}}$ are also independent of $r_{0}$, so that $\underline{\zeta}_{t,r_{0}}\rightarrow\underline{\zeta}_{t}$ on compact subsets as $r_{0}\rightarrow\infty$, where $\underline{\zeta}_{t}>0$ satisfies equation \eqref{126} with $\overline{v}_{t,r_{0}}$ replaced by $\overline{v}_{t}$. Therefore, since $\underline{\zeta}_{t,r_{0}}^{\frac{1}{3}}\leq\phi_{t,r_{0}}$, the function $\phi_{t,r_{0}}$ is uniformly bounded below by a positive constant on any compact subset, independent of $r_{0}$. As the $C^{0}$ and $C^{1}$ bounds of Proposition \ref{supbounds} and Corollary \ref{globalgrad} are also independent of $r_{0}$, $C^{2,\alpha}$ estimates on compact subsets may be established analogously to \eqref{138.1} and independently of $r_{0}$. This implies that, after possibly passing to a subsequence, $\phi_{t,r_{0}}\rightarrow\phi_{t}$ as $r_{0}\rightarrow\infty$, where $\phi_{t}$ satisfies \eqref{1.2}. Moreover since
$\underline{\zeta}_{t}^{\frac{1}{3}}\leq\phi_{t}\leq\overline{\phi}_{t}$, \eqref{1.3} is also valid and $\phi_{t}$ is strictly positive.
\end{proof}

\begin{remark}
It is unlikely that higher order regularity better than $C^{2,\alpha}$ is possible, due to the presence of an absolute value on the right-hand side of \eqref{1.2}.
\end{remark}

\section{Proof of the Main Theorem}
\label{sec7}

The purpose of this section is to prove Theorem \ref{main:theorem}. First observe that if $\rho=|q|$, or equivalently $|\partial M|_{g}=4\pi q^{2}$, then the Penrose inequality with charge \eqref{charged-penrose-inequality} is equivalent to the positive mass theorem with charge $m\geq|q|$, which has already been established in \cite{GibbonsHawkingHorowitzPerry}. The case of equality for this theorem, asserts that the initial data must be isometric to the canonical slice of the Majumdar-Papapetrou spacetime \cite{ChruscielReallTod}. However, such initial data does not possess a minimal surface boundary, and hence does not fall under the hypotheses of Theorem \ref{main:theorem}. Hence, when the area/charge inequality is saturated we must have $m>|q|$.

Now assume that $|\partial M|_{g}>4\pi q^{2}$, so that in particular Theorem \ref{exhaustion} applies.
Note that the results of Section \ref{sec4} require a perturbation of the initial data to achieve charged harmonic asymptotics, thus we assume here that such a perturbation has been made.
Theorems \ref{flowproperties} and \ref{monotonicity} imply that the areas $|\partial M_{t}|_{g_{t}}$ and charges $q_{t}$ remain constant throughout
the flow, while the mass $m(t)$ is nonincreasing. Moreover, Theorem \ref{exhaustion} guarantees that for
some finite time $\overline{t}$ the minimal boundary $\partial M_{\overline{t}}$ is connected. Consider
now the initial data set $(M_{\overline{t}},g_{\overline{t}},E_{\overline{t}},B_{\overline{t}})$. It satisfies all the hypotheses of the Penrose inequality with charge for a single black hole, namely the
boundary is an outermost minimal surface, the charged dominant energy condition is valid, and the Maxwell
constraints holds without charged matter. As the inequality has been proven for a single black hole \cite{HuiskenIlmanen}, \cite{Jang}, it follows that
\begin{equation}
m(0)\geq m(\overline{t})\geq\sqrt{\frac{|\partial M_{\overline{t}}|_{g_{t}}}{16\pi}}
+\sqrt{\frac{\pi}{|\partial M_{\overline{t}}|_{g_{t}}}}q_{\overline{t}}^{2}
=\sqrt{\frac{|\partial M_{0}|_{g_{0}}}{16\pi}}
+\sqrt{\frac{\pi}{|\partial M_{0}|_{g_{0}}}}q_{0}^{2}.
\end{equation}
Since this holds for an arbitrarily small perturbation of the data, \eqref{charged-penrose-inequality} holds for the original initial data.
Notice that we avoid the convergence issue concerning the flow, by employing the inverse mean curvature
flow once the minimal surface becomes connected. Whether the charged conformal flow converges to the
canonical Reissner-Nordstr\"{o}m data is an interesting question, which we strongly believe has an
affirmative answer. Ultimately though, it is not necessary for the current result, and so it will be
left for future investigation.

It remains to establish the rigidity statement, whose proof is outlined as follows. First, it will be shown that saturation of the Penrose inequality with charge forces the mass to be constant along the charged conformal flow. From the first derivative formula, this implies that the mass of the doubled manifold is arbitrarily small, and hence the doubled manifold must have trivial topology. In this case the original initial data has only a single boundary component, and previous results \cite{DisconziKhuri} then apply to yield the desired conclusion. We now provide a detailed account.

Here it is not necessary to perturb the initial data to achieve charged harmonic asymptotics. Suppose that equality holds in \eqref{charged-penrose-inequality}. We conclude that the mass remains constant throughout
the flow $m'(t)=0$. To see this, suppose not, then
\begin{equation}
m(\widetilde{t})<\sqrt{\frac{|\partial M_{0}|_{g_{0}}}{16\pi}}
+\sqrt{\frac{\pi}{|\partial M_{0}|_{g_{0}}}}q_{0}^{2}
\end{equation}
for some $\widetilde{t}$. A contradiction is obtained by applying the results of the previous paragraph to the initial data set $(M_{\widetilde{t}},g_{\widetilde{t}},E_{\widetilde{t}},B_{\widetilde{t}})$.
Equation \eqref{mass1} now implies that
\begin{equation}\label{270.0}
e^{-2t}(\gamma_{t}-\overline{\gamma}_{t})+\frac{1}{2}\theta_{t}^{\varepsilon}=\widetilde{m}^{\varepsilon}(t)\text{ }\text{ }\text{ }\text{ for all }\text{ }\text{ }\text{ }t\geq 0,
\end{equation}
where $\widetilde{m}^{\varepsilon}(t)$ is the mass of the doubled manifold $(M_{t}^{+}\cup M_{t}^{-},(g_{t}^{\varepsilon})^{+}\cup (g_{t}^{\varepsilon})^{-})$ and $g_{t}^{\varepsilon}$
is defined in the proof of Theorem \ref{monotonicity}.

We claim that the doubled manifold must be diffeomorphic to $\mathbb{R}^{3}$. To see this, suppose that it is not true. Then $\partial M_{t}$ must have two or more components.
According to the result of Meeks-Simon-Yau \cite{MeeksSimonYau}, there is then an outermost minimal
surface $\mathcal{S}_{t}(\varepsilon,\tau_{0})$ in the doubled manifold which encloses the nontrivial topology. 
Since the scalar curvature of the doubled manifold is nonnegative, the Penrose inequality and \eqref{270.0}
then yield
\begin{equation}\label{270}
e^{-2t}(\gamma_{t}-\overline{\gamma}_{t})+\frac{1}{2}\theta_{t}^{\varepsilon}
\geq\sqrt{\frac{|\mathcal{S}_{t}(\varepsilon,\tau_{0})|_{(g_{t}^{\varepsilon})^{+}\cup (g_{t}^{\varepsilon})^{-}}}{16\pi}}\text{ }\text{ }\text{ }\text{ for all }\text{ }\text{ }\text{ }t\geq 0.
\end{equation}
For each fixed $t$, Theorem \ref{thm15} shows that $|\gamma_{t}-\overline{\gamma}_{t}|\rightarrow 0$ as
$\tau_{0}\rightarrow 0$, and from the proof of Theorem \ref{monotonicity} we have $\theta_{t}^{\varepsilon}\rightarrow 0$ as $\varepsilon\rightarrow 0$. On the other hand $\phi_{t}\rightarrow 0$ and $\overline{v}_{t}\rightarrow v_{t}$ both in $C^{0}$ as $\tau_{0}\rightarrow 0$ and $\varepsilon\rightarrow 0$ (note that higher order convergence of $\phi_{t}$ is not generally possible), so that the conformal factors defining the doubled metric converge as $w_{t}^{\pm}\rightarrow
\frac{1\pm v_{t}}{2}$ in $C^{0}$. Let $\widetilde{g}_{t}^{\pm}=\left(\frac{1\pm v_{t}}{2}\right)^{4}g_{t}$. It follows that
\begin{equation}\label{270.001}
|\mathcal{S}_{t}(\varepsilon,\tau_{0})|_{(g_{t}^{\varepsilon})^{+}\cup (g_{t}^{\varepsilon})^{-}}
\geq\frac{1}{2}|\mathcal{S}_{t}(\varepsilon,\tau_{0})|_{\widetilde{g}_{t}^{+}\cup \widetilde{g}_{t}^{-}},
\end{equation}
for $\tau_{0}$, $\varepsilon$ sufficiently small. There is, however, a positive lower bound for the area of any surface enclosing the nontrivial topology in $(M_{t}^{+}\cup M_{t}^{-},\widetilde{g}_{t}^{+}\cup \widetilde{g}_{t}^{-})$. Hence
there is a positive lower bound independent of $\tau_{0}$ and $\varepsilon$ for the right-hand side of \eqref{270.001}. This leads to a contradiction with \eqref{270} for sufficiently small $\tau_{0}$ and $\varepsilon$. Therefore the doubled manifold must have trivial topology, or equivalently $\partial M_{t}$ consists of one component for all $t\geq 0$.

Rigidity for the Penrose inequality with charge in the case of one black hole was established in \cite{DisconziKhuri}. This result relies on monotonicity of the so called charged Hawking mass
\begin{equation}\label{chargedHM}
M_{CH}(S) = \sqrt{\frac{|S|}{16\pi} } \left(1 + \frac{4\pi q^2}{|S|} - \frac{1}{16\pi}\int_S H^2 dA\right),
\end{equation}
under inverse mean curvature flow. In \cite{DisconziKhuri} only the electric field was present, and $q$
in \eqref{chargedHM} represented the total electric charge. However, the same proof applies when both
the electric and magnetic fields are present, if $q^2=q_{e}^{2}+q_{b}^{2}$. It follows that $(M,g)$ is
isometric to the canonical slice of the Reissner-Nordstr\"{o}m spacetime, and $E=-q_{e}\nabla r^{-1}$,
$B=-q_{b}\nabla r^{-1}$ in the usual anisotropic coordinates. This concludes the proof.

\appendix

\section{Estimates at $S_{\tau_{0}}$}
\label{sec8} 

In this section we will establish the estimates for $\phi_{t}$ at $S_{\tau_{0}}$ appearing in \eqref{216.3}. For simplicity, the subindex $t$ will be suppressed. Note that by the Hopf lemma $\partial_{\tau}\phi|_{S_{\tau_{0}}}<0$, so the level sets of $\phi$ foliate a neighborhood of $S_{\tau_{0}}$. Consider the domain
\begin{equation}\label{223}
\Omega_{\rho}=\{x\in M\mid\sigma\lambda\tau_{0}^{4}<\phi<\lambda\tau_{0}^{4}\},
\end{equation}
where $\sigma$ is sufficiently small to guarantee that $\Omega_{\sigma}\subset D(\tau_{0},\frac{5}{4}\tau_{0})$, the domain enclosed by $S_{\tau_{0}}$ and $S_{\frac{5}{4}\tau_{0}}$. Then $f=\lambda^{2}[1-2\tau_{0}^{-1}(\tau-\tau_{0})]$ on $\Omega_{\sigma}$. We may then apply the method of proof for Proposition \ref{interiorgrad} to obtain a $C^{1}$ estimate on $\Omega_{\sigma}$. The key inequality which leads to the desired estimate is the analogue of \eqref{17}. Namely using \eqref{asdfgh}, $c^{-1}\tau_{0}\leq|\overline{v}|\leq c\tau_{0}$, and $s=1$, $\varepsilon=0$ we find that at a global maximum for $|\nabla \phi|$,
\begin{equation}\label{224}
0\geq\frac{c^{-1}(\Lambda-8)}{\tau_{0}^{2}\phi^{2}}|\nabla\phi|^{3}
-\frac{c\Lambda}{\tau_{0}^{3}\phi}|\nabla\phi|^{2}
-c\left(\frac{\Lambda f}{\phi^{2}}+1\right)|\nabla\phi|-\frac{\Lambda}{\phi}|\nabla f|-\frac{c\Lambda f}{\tau_{0}\phi}
\end{equation}
for some constant $c>0$ independent of $\lambda$, $\tau_{0}$, and $\sigma$. It is clear that if $|\nabla\phi|>c_{0}\lambda\tau_{0}$
for $c_{0}$ sufficiently large, then a contradiction is obtained from \eqref{224}. We conclude that at a
global interior maximum, there exists a finite constant $c$ such that $|\nabla\phi|\leq c\lambda\tau_{0}$.
It follows that
\begin{equation}\label{230}
\sup_{\Omega_{\sigma}}|\nabla \phi|\leq c\lambda\tau_{0}+\sup_{\partial\Omega_{\sigma}}|\nabla \phi|.
\end{equation}

The boundary of $\Omega_{\sigma}$ consists of two types of components. One is $S_{\tau_{0}}$, where a gradient estimate has already been established in Lemma \ref{lemmagrad}. The other type of component is also a level set of $\phi$, and possesses a neighborhood in which $f$ has the desired expression, so Lemma \ref{lemmagrad} applies here as well to yield $|\nabla\phi|_{\partial\Omega_{\sigma}}\leq c\lambda\tau_{0}$. Together with \eqref{230}, this implies that
\begin{equation}\label{233}
\sup_{\Omega_{\sigma}}|\nabla\phi|\leq c\lambda\tau_{0}.
\end{equation}

We now claim that
$S_{\overline{\tau}_{0}}\subset\Omega_{\sigma}$ for some $\overline{\tau}_{0}>\tau_{0}$ independent of $\lambda$.
To see this, let $\overline{x}$ be coordinates on $S_{\tau}$, and observe that \eqref{233} implies
\begin{equation}\label{236}
\phi(\overline{x},\tau)-\phi(\overline{x},\tau_{0})
=\int_{\tau_{0}}^{\tau}\partial_{\tau}\phi(\overline{x},\varsigma)d\varsigma\geq
-c\lambda\tau_{0}(\tau-\tau_{0}),
\end{equation}
so that
\begin{equation}\label{237}
\phi(\overline{x},\tau)\geq\lambda\tau_{0}^{4}-c\lambda\tau_{0}(\tau-\tau_{0}).
\end{equation}
Thus, if $\overline{\tau}_{0}=\tau_{0}+\frac{1}{2c}\tau_{0}^{3}$ then
\begin{equation}\label{238}
\phi|_{S_{\overline{\tau}_{0}}}\geq\frac{1}{2}\lambda\tau_{0}^{4}>\sigma\lambda\tau_{0}^{4},
\end{equation}
assuming that $\sigma<1/2$. This verifies the claim. It follows that $\operatorname{Vol}(\Omega_{\sigma})\geq C_{0}>0$
where $C_{0}$ is independent of $\lambda$. Therefore
the constants appearing in the $L^{p}$ and Schauder estimates, as well as the Sobolev embeddings, are independent of $\lambda$.

By the $L^{p}$ estimates and the equation \eqref{1.2} satisfied by $\phi$,
\begin{equation}\label{240}
\parallel\phi\parallel_{W^{2,p}(\Omega_{\sigma})}
\leq c\left(\Lambda\parallel \phi^{-1}f\parallel_{L^{p}(\Omega_{\sigma})}+\Lambda\parallel\phi^{-1}|\nabla\phi|^{2}\parallel_{L^{p}(\Omega_{\sigma})}
+\parallel\phi\parallel_{L^{p}(\Omega_{\sigma})}+\parallel\phi\parallel_{W^{2,p}(\partial\Omega_{\sigma})}\right)
\end{equation}
where the constant $c$ depends on $\tau_{0}$ but not on $\lambda$. With the aid of \eqref{233}, $\sigma\lambda\tau_{0}^{4}\leq\phi\leq\lambda\tau_{0}^{4}$, and $f\leq \lambda^{2}$, this implies that
\begin{equation}\label{241}
\parallel\phi\parallel_{W^{2,p}(\Omega_{\sigma})}\leq c(\tau_{0})\lambda,\text{ }\text{ }\text{ }\text{ }\text{ }\text{ which yields }\text{ }\text{ }\text{ }\text{ }\text{ }\text{ }|\phi|_{C^{1,\alpha}(\Omega_{\sigma})}\leq c(\tau_{0})\lambda.
\end{equation}
By the Schauder estimates
\begin{equation}\label{243}
|\phi|_{C^{2,\alpha}(\Omega_{\sigma})}
\leq c\left(\Lambda| \phi^{-1}f|_{C^{0,\alpha}(\Omega_{\sigma})}+\Lambda|\phi^{-1}|\nabla\phi|^{2}|_{C^{0,\alpha}(\Omega_{\sigma})}
+|\phi|_{C^{0}(\Omega_{\sigma})}+|\phi|_{C^{2,\alpha}(\partial\Omega_{\sigma})}\right)\leq c(\tau_{0})\lambda.
\end{equation}
Due to the absolute value on the right-hand side of \eqref{1.2}, $C^{2,\alpha}$-estimates are generally the highest order estimates possible. However, since one-sided derivatives of the absolute value of a smooth
function are Lipschitz, taking the one-sided derivative $\partial_{\tau}$ at $S_{\tau_{0}}$ yields an equation
for $\partial_{\tau}\phi$ whose right-hand side is Lipschitz near $S_{\tau_{0}}$. Then $C^{2,\alpha}$-estimates follow for $\partial_{\tau}\phi$. We have thus proven the following theorem.

\begin{theorem}\label{thm14}
If $\tau_{0}$ is sufficiently small and $\Lambda>8$, then the solution of \eqref{1.2}, \eqref{1.3} constructed in Theorem \ref{existence1} satisfies
\begin{equation}\label{244}
|\nabla^{2}\partial_{\tau}\phi_{t}|_{S_{\tau_{0}}}+|\nabla^{2}\phi_{t}|_{S_{\tau_{0}}}
+|\nabla\phi_{t}|_{S_{\tau_{0}}}
+\phi_{t}|_{S_{\tau_{0}}}\leq c(\tau_{0})\lambda,
\end{equation}
where $c(\tau_{0})$ is independent of $\lambda$.
\end{theorem}

\section{The Reissner-Nordstr\"{o}m Flow}
\label{sec9} 

In this section we construct the charged conformal flow in the canonical slice of the Reissner-Nordstr\"{o}m spacetime. Consider the exterior Reissner-Nordstr\"{o}m metric with mass $m$ and squared total charge $q^2=q_{e}^{2}+q_{b}^{2}$, in non-isotropic coordinates,
\begin{equation}\label{5.1}
-\left(1-\frac{2m}{\overline{r}}+\frac{q^{2}}{\overline{r}^{2}}\right)d\overline{t}^{2}
+\left(1-\frac{2m}{\overline{r}}+\frac{q^{2}}{\overline{r}^{2}}\right)^{-1}d\overline{r}^{2}
+\overline{r}^{2}d\sigma^{2},\text{ }\text{ }\text{ }\text{ }\text{ }\text{ }
\overline{r}\geq m+\sqrt{m^{2}-q^{2}},
\end{equation}
where $d\sigma^{2}$ is the round metric on the 2-sphere. The electric and magnetic fields are given by
\begin{equation}\label{5.1'}
E_{i}=-\left(1-\frac{2m}{\overline{r}}+\frac{q^{2}}{\overline{r}^{2}}\right)^{-1}
\partial_{i}\left(\frac{q_{e}}{\overline{r}}\right),\text{ }\text{ }\text{ }\text{ }\text{ }
B_{i}=-\left(1-\frac{2m}{\overline{r}}+\frac{q^{2}}{\overline{r}^{2}}\right)^{-1}
\partial_{i}\left(\frac{q_{b}}{\overline{r}}\right).
\end{equation}
Change coordinates by
\begin{equation}\label{5.2}
\overline{r}=r+m+\frac{m^{2}-q^{2}}{4r}
\end{equation}
to obtain the expression of the metric in isotropic coordinates
\begin{equation}\label{5.3}
-\left(\frac{1-\frac{m^{2}-q^{2}}{4r^{2}}}{1+\frac{m}{r}+\frac{m^{2}-q^{2}}{4r^{2}}}\right)^{2}d\overline{t}^{2}
+\left(1+\frac{m}{r}+\frac{m^{2}-q^{2}}{4r^{2}}\right)^{2}\delta,\text{ }\text{ }\text{ }\text{ }\text{ }\text{ }
r\geq\frac{\sqrt{m^{2}-q^{2}}}{2},
\end{equation}
with corresponding electric and magnetic fields fields
\begin{equation}\label{5.4}
E_{i}=-\left(1+\frac{m}{r}+\frac{m^{2}-q^{2}}{4r^{2}}\right)^{-1}\partial_{i}\left(\frac{q_{e}}{r}\right),\text{ }\text{ }\text{ }\text{ }\text{ }\text{ }
B_{i}=-\left(1+\frac{m}{r}+\frac{m^{2}-q^{2}}{4r^{2}}\right)^{-1}\partial_{i}\left(\frac{q_{b}}{r}\right),
\end{equation}
where $\delta$ is the Euclidean metric.

We may write the Reissner-Nordstr\"{o}m spacetime metric as $-V^{2}d\overline{t}^{2}+g$. In isotropic coordinates $g=U^{4}\delta$, where
\begin{equation}\label{5.5}
U(x)=\sqrt{1+\frac{m}{r}+\frac{m^{2}-q^{2}}{4r^{2}}}.
\end{equation}
The electric field may now be expressed as
\begin{equation}\label{5.6}
E_{i}=-V^{-1}\partial_{i}\left(\frac{q_{e}}{\rho}\right)\text{ }\text{ }\text{ in non-isotropic coordinates, }\text{ }\text{ }
E_{i}=-U^{-2}\partial_{i}\left(\frac{q_{e}}{r}\right)\text{ }\text{ }\text{ in isotropic coordinates}.
\end{equation}
Notice that this makes sense from previous formulas, since (in isotropic coordinates) we know that $E^{i}=U^{-6}E_{\delta}^{i}$ is divergence free whenever $E_{\delta}$ is divergence free with respect to $\delta$; this is
of course the case, as $E_{\delta}=q_{e}\nabla r^{-1}$. We can also check that the electric fields agree in the two different coordinates:
\begin{align}\label{5.7.0}
\begin{split}
E_{\overline{r}}d\overline{r} &= \frac{1}{V}\frac{q_{e}}{\overline{r}^{2}}d\overline{r}
=\frac{1}{V}\frac{q_{e}}{\overline{r}^{2}}\frac{d\overline{r}}{dr}dr
=\frac{1}{V}\frac{q_{e}}{\overline{r}^{2}}\left(1-\frac{m^{2}-q^{2}}{4r^{2}}\right)dr\\
&=\frac{1}{V}\frac{q_{e}}{r^{2}U^{4}}\left(1-\frac{m^{2}-q^{2}}{4r^{2}}\right)dr
=-U^{-2}\partial_{r}\left(\frac{q_{e}}{r}\right)dr=E_{r}dr.
\end{split}
\end{align}
Similar considerations hold for the magnetic field.

The conformal flow $g_{t}=u_{t}^{4}g$ is given by rescaling coordinates by $x\mapsto e^{-2t}x$, thus
\begin{equation}\label{5.7}
u_{t}(x)=\frac{\sqrt{e^{-2t}+\frac{m}{r}+\frac{e^{2t}(m^{2}-q^{2})}{4r^{2}}}}{\sqrt{1+\frac{m}{r}+\frac{m^{2}-q^{2}}{4r^{2}}}}.
\end{equation}
In order to calculate the flow velocity $v_{t}$, observe that
\begin{equation}\label{5.8}
v_{t}=\frac{d}{dt}\log u_{t}=\frac{-e^{-2t}+e^{2t}\frac{m^{2}-q^{2}}{4r^{2}}}{e^{-2t}+\frac{m}{r}+e^{2t}\frac{m^{2}-q^{2}}{4r^{2}}}.
\end{equation}
As usual $v_{t}=0$ on the minimal surface $\partial M_{t}=\{r=\frac{\sqrt{m^{2}-q^{2}}}{2}e^{2t}\}$, and $v_{t}\rightarrow -1$ as $r\rightarrow\infty$. Set $U_{t}=u_{t}U$, then $g_{t}=U_{t}^{4}\delta$.
In order to calculate the equation satisfied by $v_{t}$, recall the identities
\begin{equation}\label{5.9}
L_{g_{t}}v_{t}=U_{t}^{-5}L_{\delta}(U_{t}v_{t}),\text{ }\text{ }\text{ }\text{ }\text{ }\text{ }R_{g_{t}}=-8U_{t}^{-5}L_{\delta}U_{t}=-8U_{t}^{-5}\Delta_{\delta}U_{t}.
\end{equation}
It follows that
\begin{align}\label{5.10}
\begin{split}
\Delta_{g_{t}}v_{t}&=\frac{1}{8}R_{g_{t}}v_{t}+U_{t}^{-5}\Delta_{\delta}(U_{t}v_{t})\\
&=-U_{t}^{-5}v_{t}\Delta_{\delta}U_{t}+U_{t}^{-5}(U_{t}\Delta_{\delta}v_{t}+2\nabla U_{t}\cdot\nabla v_{t}+v_{t}\Delta_{\delta}U_{t})\\
&=U_{t}^{-4}(\Delta_{\delta}v_{t}+2\nabla\log U_{t}\cdot\nabla v_{t}).
\end{split}
\end{align}
A computation shows that
\begin{align}\label{5.11}
\begin{split}
\Delta_{\delta}v_{t}&=\partial_{r}^{2}v_{t}+\frac{2}{r}\partial_{r}v_{t}\\
&=\frac{\frac{m^{2}-q^{2}}{r^{4}}+e^{2t}\frac{m(m^{2}-q^{2})}{2r^{5}}}{\left(e^{-2t}+\frac{m}{r}+e^{2t}\frac{m^{2}-q^{2}}{4r^{2}}\right)^{2}}
-\frac{2\left(e^{-2t}\frac{m}{r^{2}}+\frac{m^{2}-q^{2}}{r^{3}}+e^{2t}\frac{m(m^{2}-q^{2})}{4r^{4}}\right)\left(\frac{m}{r^{2}}+e^{2t}\frac{m^{2}-q^{2}}{2r^{3}}\right)}
{\left(e^{-2t}+\frac{m}{r}+e^{2t}\frac{m^{2}-q^{2}}{4r^{2}}\right)^{3}},
\end{split}
\end{align}
and
\begin{equation}\label{5.12}
\partial_{r}v_{t}=-\frac{e^{-2t}\frac{m}{r^{2}}+\frac{m^{2}-q^{2}}{r^{3}}+e^{2t}\frac{m(m^{2}-q^{2})}{4r^{2}}}{\left(e^{-2t}+\frac{m}{r}+e^{2t}\frac{m^{2}-q^{2}}{4r^{2}}\right)^{2}},
\end{equation}
\begin{equation}\label{5.13}
\partial_{r}\log U_{t}=-\frac{\frac{m}{r^{2}}+e^{2t}\frac{m^{2}-q^{2}}{2r^{3}}}{2\left(e^{-2t}+\frac{m}{r}+e^{2t}\frac{m^{2}-q^{2}}{4r^{2}}\right)}.
\end{equation}
Therefore
\begin{equation}\label{5.14}
\Delta_{g_{t}}v_{t}=\frac{q^{2}r^{-4}}{\left(e^{-2t}+\frac{m}{r}+e^{2t}\frac{m^{2}-q^{2}}{4r^{2}}\right)^{4}}v_{t}.
\end{equation}
However since $E_{t}^{i}=U_{t}^{-6}E_{\delta}^{i}$ we have
\begin{equation}\label{5.15}
|E_{t}|_{g_{t}}^{2}=U_{t}^{-8}|E_{\delta}|_{\delta}^{2}
=\frac{q_{e}^{2}r^{-4}}{\left(e^{-2t}+\frac{m}{r}+e^{2t}\frac{m^{2}-q^{2}}{4r^{2}}\right)^{4}},
\end{equation}
with an identical formula for the squared norm of the magnetic field when $q_{e}$ is replaced by $q_{b}$.
It is also true that $R_{g_{t}}=2\left(|E_{t}|_{g_{t}}^{2}+|B_{t}|_{g_{t}}^{2}\right)$, which may be verified from the identity
\begin{equation}\label{5.16}
R_{g_{t}}=-8U_{t}^{-5}L_{\delta}U_{t}=-8U_{t}^{-5}\Delta_{\delta}U_{t}.
\end{equation}
Thus, the equation satisfied by the velocity function is
\begin{equation}\label{5.17}
\Delta_{g_{t}}v_{t}-\left(|E_{t}|_{g_{t}}^{2}+|B_{t}|_{g_{t}}^{2}\right)v_{t}=0\text{ }\text{ }\text{ }\text{ or }\text{ }\text{ }\text{ }\Delta_{g_{t}}v_{t}-\frac{1}{2}R_{g_{t}}v_{t}=0.
\end{equation}

We remark that this equation for $v_{t}$ is precisely the one satisfied by the warping factor of the Reissner-Nordstr\"{o}m spacetime. In fact the Reissner-Nordstr\"{o}m spacetime
metric may be written as $-v_{t}^{2}d\overline{t}^{2}+u_{t}^{4}g$.  This structure is also valid for the Schwarzschild spacetime
if $u_{t}$ and $v_{t}$ arise from Bray's original conformal flow. Given an arbitrary metric $g$, this shows how to associate a static spacetime with $(M,g)$, namely use the conformal factor
and velocity functions from the conformal flow, to generate a static spacetime as above.

\end{document}